\documentclass{article}

\usepackage{hyperref}
\usepackage{amsmath}
\usepackage{amsthm}
\usepackage{amssymb}

\usepackage[utf8]{inputenc}

\usepackage{modello} 
\usepackage{tikz}
\usetikzlibrary{calc,
                positioning}

\begin{document}

\title{Dual and Hull code in the first two generic constructions and relationship with the Walsh transform of cryptographic functions}

\author{Virginio Fratianni}
\date{Université Paris 8, Laboratoire de Géométrie, Analyse et Applications, LAGA, Université Sorbonne Paris Nord, CNRS, UMR 7539, France.\\ \emph{Email address}: virginio.fratianni@etud.univ-paris8.fr}

\maketitle            



\section*{Abstract}
We contribute to the knowledge of linear codes from special polynomials and functions, which have been studied intensively in the past few years. Such codes have several applications in secret sharing, authentication codes, association schemes and strongly regular graphs.

This is the first work in which we study the dual codes in the framework of the two generic constructions; in particular, we propose a Gram-Schmidt (complexity of $\mathcal{O}(n^3)$) process to compute them explicitly. The originality of this contribution is in the study of the existence or not of defining sets $D'$, which can be used as ingredients to construct the dual code $\mathcal{C}'$ for a given code $\mathcal{C}$ in the context of the second generic construction. We also determine a necessary condition expressed by employing the Walsh transform for a codeword of  $\mathcal{C}$  to belong in the dual. This achievement was done in general and when the involved functions are weakly regularly bent. We shall give a novel description of the Hull code in the framework of the two generic constructions. Our primary interest is constructing linear codes of fixed Hull dimension and determining the (Hamming)  weight of the codewords in their duals. 
\\\\
\textbf{Keywords} Linear codes, dual code, Hull code, $p$-ary codes, weight distribution, $p$-ary functions, bent functions, weakly regular bent functions, vectorial functions, cyclotomic fields, secret sharing schemes, Walsh transform, character theory.

\section{Introduction}
Coding theory is the study of the properties of codes and their fitness for specific applications, like data compression for cloud storage, wireless data transmission and especially cryptography. It involves various scientific disciplines, for example, mathematics, information theory, computer science, artificial intelligence and linguistics; however, it is based on mathematical methods, especially algebraic ones in the linear code case. For instance, we have many connections with number theory in the context of Gauss sums, cyclotomic field theory, Galois fields and tools from the theory of Fourier transform and the related exponential sums.
\\\\
Nowadays, much importance is attached to coding theory because many infrastructures rely on codes, like wireless communications and cloud storage. To a great extent, the research
on codes depends on mathematics since nearly all the theoretical basis, constructions, and analysis of codes are based on mathematical methods.
\\\\
The focus of this paper is to deepen the relation between coding theory and symmetric cryptography in the context of the codes built by means of cryptographic functions. In particular, the center of attention is the dual code in these two constructions, for which previously there was no literature, and the relation with the fundamental concept of the Walsh transform of a cryptographic function. For instance, in \cite{WeightWalshTransform} Mesnager established a deep connection between the Walsh transform and the weight enumerator in the first generic construction. In this paper, we would like to follow this path in the context of the dual codes. 

\subsection{Some background related to coding theory}
An $[n, k, d]$ linear code $\mathcal{C}$ over $\mathbb{F}_{q}$ is a linear subspace of $\mathbb{F}_{q}^n$ with dimension $k$ and minimum (Hamming) distance $d$. The value $n-k$ is called the codimension of $\mathcal{C}$. Given a linear code $\mathcal{C}$ of length $n$ over $\mathbb{F}_{q}$ (resp. $\mathbb{F}_{q^{2}}$ ), its Euclidean dual code (resp. Hermitian dual code) is denoted by $\mathcal{C}^{\perp}$ (resp. $\left.\mathcal{C}^{\perp_{H}}\right)$. The codes $\mathcal{C}^{\perp}$ and $\mathcal{C}^{\perp_{H}}$ are defined by
$$
\mathcal{C}^{\perp}=\left\{\left(b_{0}, b_{1}, \ldots, b_{n-1}\right) \in \mathbb{F}_{q}^{n}\,:\, \sum_{i=0}^{n-1} b_{i} c_{i}=0 \,\,\text{for every}\,\left(c_{0}, c_{1}, \ldots, c_{n-1}\right) \in \mathcal{C}\right\};
$$
$$
\mathcal{C}^{\perp_H}=\left\{\left(b_{0}, b_{1}, \ldots, b_{n-1}\right) \in \mathbb{F}_{q^2}^{n}\,:\, \sum_{i=0}^{n-1} b_{i} \overline{c_{i}}=0 \,\,\text{for every}\,\left(c_{0}, c_{1}, \ldots, c_{n-1}\right) \in \mathcal{C}\right\}
$$
respectively.
The minimum distance of an $[n, k, d]$ linear code is bounded by the Singleton bound
$$
d \leq n+1-k .
$$
A code meeting the above bound is called Maximum Distance Separable (MDS).

\subsection{General theory of functions over finite fields}
     In this chapter, we shall identify the Galois field $\mathbb{F}_{p^n}$ of order $p^n$ with the vector space  $\mathbb{F}_{p}^n$   by using $\mathbb{F}_{p^n}$ as an $n$-dimensional vector space over $\mathbb{F}_p$ and writing $\mathbb{F}_{p}^n$ instead of $\mathbb{F}_{p^n}$ when the field structure is not really used.
 In the following, $``\cdot"$ denotes the standard inner (dot) product of two vectors, that is, $\lambda \cdot x:=\lambda_1x_1+ \ldots + \lambda_nx_n$, where $\lambda,x\in \mathbb{F}_p^n$.\\\\
 The trace function $\text{Tr}_{p^n/p}:\mathbb{F}_{p^n}\rightarrow\mathbb{F}_p$ is defined as 
 $$\text{Tr}_{p^n/p}(x):=\sum_{i=0}^{n-1}x^{p^i}=x+x^p+x^{p^2}+\dots+x^{p^{n-1}}.$$
 If we identify the vector space $\mathbb{F}_p^n$ with
the finite field $\mathbb{F}_{p^n}$, we use the trace bilinear form
$Tr_{p^n/p}(\lambda x)$  instead of the dot product,  that is, 
$$\lambda \cdot x=Tr_{p^n/p}(\lambda x),$$
where $\lambda,x\in \mathbb{F}_{p^n}$ and ``$Tr_{p^n/p}$" denotes the absolute trace over $ \mathbb{F}_{p^n}$.

\begin{defn}
Let $q=p^r$ where $p$ is a prime. A vectorial function
 $\mathbb{F}_q^n\rightarrow \mathbb{F}_q^m$ (or $\mathbb{F}_{q^n}\rightarrow \mathbb{F}_{q^m}$) is called an \emph{$(n,m)$-q-ary function}. When $q=2$, an $(n,m)$-2-ary function will be simply denoted an \emph{$(n,m)$-function}. A \emph{Boolean function} is an $(n,1)$-function, i.e. a function $\mathbb{F}_2^n\rightarrow \mathbb{F}_2$ (or $\mathbb{F}_{2^n}\rightarrow \mathbb{F}_2$).
\end{defn}

Let $f:\mathbb{F}_{p^m}\rightarrow\mathbb{F}_p$ be a $p$-ary function, then the Walsh transform of $f$ is given by:
$$\hat{\chi}_f(\lambda)=\sum_{x\in\mathbb{F}_{p^m}}\zeta_p^{f(x)-\text{Tr}_{p^m/p}(\lambda x)},$$
for $\lambda\in\mathbb{F}_{p^m}$; where $\zeta=e^{\frac{2\pi i}{p}}$ is a complex primitive $p$-th root of the unity.\\

As we can see in \cite{WeightWalshTransform}, a $p$-ary function $f: \mathbb{F}_{p^m} \longrightarrow \mathbb{F}_p$ is called bent if all its Walsh-Hadamard coefficients satisfy $\left|\hat{\chi}_f(b)\right|^2=p^m$. 
A bent function $f$ is called regular bent if, for every $b \in \mathbb{F}_{p^m}$, 
$$p^{-\frac{m}{2}} \hat{\chi}_f(b)=\zeta_p^{f^*(b)}$$ 
for some $p$-ary function $f^*: \mathbb{F}_{p^m} \rightarrow \mathbb{F}_p$. The bent function $f$ is called weakly regular bent if there exists a complex number $u$ with $|u|=1$ and a $p$-ary function $f^*$ such that 
$$u p^{-\frac{m}{2}} \widehat{\chi}_f(b)=\zeta_p^{f^*(b)}$$ 
for all $b \in \mathbb{F}_{p^m}$. Such function $f^*(x)$ is called the dual of $f(x)$. A weakly regular bent function $f(x)$ satisfies
$$
\hat{\chi}_f(b)=\epsilon \sqrt{p^m} \xi_p^{f^*(b)}
$$
where $\epsilon= \pm 1$ is called the sign of the Walsh transform of $f(x)$ and $p^*$ denotes $\left(\frac{-1}{p}\right) p$. Also, for a weakly regular bent function $f(x)$, we have 
$$\hat{\chi}_{f^*}(x)=\sum_{b \in \mathbb{F}_{p^m}} \xi_p^{f^*(b)-T r_{p^m / p}(b x)}=\epsilon p^m \zeta_p^{f(x)} / \sqrt{p^{*}}^m.$$
Moreover, Walsh-Hadamard transform coefficients of a p-ary bent function $f$ with odd $p$ satisfy
$$
p^{-\frac{m}{2}} \hat{\chi}_f(b)= \begin{cases} \pm \zeta_p^{f^*(b)}, & \text { if } m \text { is even or } m \text { is odd and } p \equiv 1 \quad(\bmod 4), \\ \pm i \zeta_p^{f^*(b)}, & \text { if } m \text { is odd and } p \equiv 3 \quad(\bmod 4).\end{cases}
$$
Hence, if $m$ is odd, then $p\equiv 1$ (mod $4$) and the constant $u$ can only be equal to $\pm 1$ or $\pm i$. 

\section{On the first generic construction}
In this section we will see a brief state of the art of the first generic construction of linear codes.

\subsection{The main construction of codes from functions $\mathcal{C}(f)$ and results on their related parameters}

    The first generic construction is obtained by considering  a code $\mathcal{C}(f)$ over $\mathbb{F}_p$ involving a polynomial $f$ from $\mathbb{F}_q$ to $\mathbb{F}_q$ (where $q=p^m$). Such a code is defined by 
$$\mathcal{C}(f)=\{{\bf c}=(Tr_{q/p}(af(x)+bx))_{x\in\mathbb{F}_q}\mid a\in\mathbb{F}_q, b\in\mathbb{F}_q\}.$$
The resulting code $\mathcal{C}(f)$ from $f$ is a linear code over $\mathbb{F}_p$ of length $q$ and its dimension is  upper bounded by $2m$ which is reached when the nonlinearity of the vectorial function $f$ is larger than $0$, which happens in many cases.

One can also define a code $\mathcal{C}^\star(f)$ over $\mathbb{F}_p$ involving a polynomial $f$ from $\mathbb{F}_q$ to $\mathbb{F}_q$ (where $q=p^m$) which vanishes at $0$ defined by 
$$\mathcal{C}^\star(f)=\{{\bf c}=(Tr_{q/p}(af(x)+bx))_{x\in\mathbb{F}^*_q}\mid a\in\mathbb{F}_q, b\in\mathbb{F}_q\}.$$
The resulting code $\mathcal{C}^\star(f)$ from $f$ is a linear code of length $q-1$ and its dimension is  also upper bounded by $2m$ which is reached in many cases.
\\\\
The first generic construction has a long history in relation with the application of $p$-ary functions in coding theory and cryptography, as it presented in \cite{DingFunctionCodes}. In particular, its importance is supported by Delsarte's theorem in the paper \cite{Delsarte} about modified Reed-Solomon codes and in the binary case it gives a coding theory characterization of APN monomials, almost bent functions and semibent functions (see for instance \cite{First1}, \cite{First2} and \cite{First3}). 
\\\\
The fundamental results on the first generic construction codes have been presented by Mesnager in \cite{WeightWalshTransform}, especially in relation with the weight enumerator of the codes thanks to a formula involving the Walsh transform which brought to many results in the literature about few weights linear codes. 
\\\\
For any $\alpha,\beta\in\mathbb{F}_{q^r}$ define 
$$
\begin{aligned}
f_{\alpha, \beta}: \mathbb{F}_{q^r} & \longrightarrow \mathbb{F}_q \\
x & \longmapsto f_{\alpha, \beta}(x):=T r_{q^r / q}(\alpha \Psi(x)-\beta x)
\end{aligned}
$$
where $\Psi$ is a mapping from $\mathbb{F}_{q^r}$ to $\mathbb{F}_{q^r}$ such that $\Psi(0)=0$. We now define a linear code $\mathcal{C}_{\Psi}$ over $\mathbb{F}_q$ as :
$$
\mathcal{C}_{\boldsymbol{\psi}}:=\left\{\bar{c}_{\alpha, \beta}=\left(f_{\alpha, \beta}\left(\zeta_1\right), f_{\alpha, \beta}\left(\zeta_2\right), \cdots, f_{\alpha, \beta}\left(\zeta_{q^r-1}\right)\right), \alpha, \beta \in \mathbb{F}_{q^r}\right\}
$$
where $\zeta_1, \cdots, \zeta_{q^r-1}$ denote the nonzero elements of $\mathbb{F}_{q^r}$.
\\\\
The following proposition shows that in many cases the codes have maximal dimension. 
\begin{prop}[\cite{WeightWalshTransform}, Proposition 1]
    The linear code $\mathcal{C}_{\Psi}$ has length $q^r-1$. If the mapping $\Psi$ has no linear component then $\mathcal{C}_{\Psi}$ has dimension $2 r$.
\end{prop}

\subsection{Direct link between Hamming weight of codewords of $\mathcal{C}(f)$ and the Walsh transform of $f$}
The following proposition makes explicit the computation of the weight of a codeword in the code. 
\begin{prop}[\cite{WeightWalshTransform}, Proposition 2]
    We keep the notation above. Let $a \in \mathbb{F}_{p^m}$. Let us denote by $\psi$ a mapping from $\mathbb{F}_{p^m}$ to $\mathbb{F}_p$ defined as:
$$
\psi_a(x)=T r_{p^m / p}(a \Psi(x))
$$
For $\bar{c}_{\alpha, \beta} \in \mathcal{C}_{\Psi}$, we have:
$$
w t\left(\bar{c}_{\alpha, \beta}\right)=p^m-\frac{1}{q} \sum_{\omega \in \mathbb{F}_q} \widehat{\chi_{\psi_{\omega \alpha}}}(\omega \beta) .
$$
\end{prop}

Thanks to the latter formula, in the same paper Mesnager presented a theorem with a new family of linear codes with few weights from weakly regular bent function and also exhibited its complete weight enumerators.
\begin{thm}[\cite{WeightWalshTransform}, Theorem 1]
     Let $\mathcal{C}$ be the previously defined linear code whose codewords are denoted by $\tilde{c}_{\alpha, \beta}$. Assume that the function $\psi_1:=T r_{p^m / p}(\Psi)$ is bent or weakly regular bent if $p=2$ or $p$ odd, respectively. We denote by $\psi_1^{\star}$ its dual function. Then the weight distribution of $\mathcal{C}$ is given as follows. In any characteristic, $w t\left(\tilde{c}_{0,0}\right)=0$ and for $\beta \neq 0, w t\left(\tilde{c}_{0, \beta}\right)=p^m-p^{m-1}$. Moreover,
     \begin{enumerate}
         \item when $p=2$ then the Hamming weight of $\tilde{c}_{1, \beta}\left(\beta \in \mathbb{F}_{2^{\mathrm{m}}}^{\star}\right)$ is given by wt $\left(\tilde{c}_{1, \beta}\right)=$ $2^{m-1}-(-1)^{\psi i}(\beta) 2^{\frac{m}{2}-1}$.
         \item when $p$ is odd then
         \begin{itemize}
             \item if $m$ is odd, the Hamming weight of $\tilde{c}_{\alpha, \beta}$ is given by
$$
\left\{\begin{array}{l}
p^m-p^{m-1} \text { if } \alpha \in \mathbb{F}_p^{\star} \text { and } \psi_1^*(\bar{\alpha} \beta)=0 ; \\
p^m-p^{m-1}-\epsilon\left(\frac{-1}{p}\right)^{\frac{m+1}{2}} p^{\frac{m-1}{2}}\left(\frac{\psi_i(\bar{\alpha} \beta)}{p}\right) \text { if } \alpha \in \mathbb{F}_p^{\star} \text { and } \psi_1^*(\bar{\alpha} \beta) \in \mathbb{F}_{p^m}^* .
\end{array}\right.
$$
    \item if $m$ is even, the Hamming weight of $\tilde{c}_{\alpha, \beta}$ is given by
$$
\left\{\begin{array}{l}
p^m-p^{m-1}-p^{\frac{m}{2}-1} \epsilon(p-1) \text { if } \alpha \in \mathbb{F}_p^{\star} \text { and } \psi_1^*(\bar{\alpha} \beta)=0 \\
p^m-p^{m-1}+p^{\frac{m}{2}-1} \epsilon \text { if } \alpha \in \mathbb{F}_p^{\star} \text { and } \psi_1^*(\bar{\alpha} \beta) \in \mathbb{F}_{p^m}^{\star} .
\end{array}\right.
$$
         \end{itemize}
     \end{enumerate}
\end{thm}
\subsection{Characterization of minimal distance of $\mathcal{C}(f)^\perp$ when $f$ is APN/PN/AB} 
In the excellent survey \cite{Survey} it is possible to find all the literature about the first generic construction codes. For example, in 1998 Carlet, Charpin and Zinoviev established the link between the dual code and APN as well as AB functions.
\begin{thm}[\cite{DualFirst}, Theorem 5]
Let $\mathcal{C}_F$ be defined as in the first generic construction, $d^{\perp}$ be the minimal distance of the dual code and $\Omega=\left\{j: A_j \neq 0,1 \leq j \leq p^m-1\right\}$ be the characteristic set of $C_F^{\perp}$, where $\left(1, A_1, \cdots, A_{p^m-1}\right)$ is the weight distribution of $\mathcal{C}_F$. If $p=2$, then
\begin{enumerate}
    \item $\mathcal{C}_F^{\perp}$ is such that $3 \leq d^{\perp} \leq 5$;
    \item $F$ is $A P N$ if and only if $d^{\perp}=5$;
    \item $F$ is $A B$ if and only if the characteristic set of $\mathcal{C}_F^{\perp}$ looks as $\left\{2^{m-1}, 2^{m-1} \pm 2^{(m-1) / 2}\right\}$.
\end{enumerate}
\end{thm}
In 2005 Carlet, Ding and Yuan in \cite{FirstDing} gave the fundamental parameters of these linear codes and their extension $\overline{\mathcal{C}_f}$, where
$$\overline{\mathcal{C}_f}:=\{(\text{Tr}(af(x)+bx+c)_{x\in\mathbb{F}_{p^m}^*}\,:\,a,b,c\in\mathbb{F}_{p^m}\}.$$
In particular we have the following bounds. 
\begin{thm}[\cite{FirstDing}, Theorems 4 and 5]
    If $F(x)$ is $P N$ with $F(0)=0$, then $\mathcal{C}_F\left(\right.$ resp. $\overline{\mathcal{C}}_F$ ) has parameters $\left[p^m-1,2 m / h, d ; p^h\right]$ (resp. $\left.\left[p^m-1,1+2 m / h, d ; p^h\right]\right)$ with
$$
d \geq \frac{p^h-1}{p^h}\left(p^m-p^{m / 2}\right) .
$$
Furthermore, for every nonzero weight $\omega$ in both $\mathcal{C}_F$ and $\overline{\mathcal{C}}_F$, we have
$$
\frac{p^h-1}{p^h}\left(p^m-p^{m / 2}\right) \leq w \leq \frac{p^h-1}{p^h}\left(p^m+p^{m / 2}\right)
$$
\end{thm}
To improve the latter bounds, they also studied the linear codes arising from specific function families. 

    \section{On the second generic construction}
    In this section we will see a brief state of the art of the second generic construction of linear codes.
    \subsection{The main construction of codes from defining sets $D$ and results on their related parameters}
     The second generic construction of linear codes from functions is obtained by fixing a set $D=\{d_1, d_2, \cdots, d_n\}$ in $\mathbb{F}_q$ (where $q=p^k$) and by defining a linear code involving $D$ as follows:
 $$\mathcal {C}_D=\{(Tr_{q/p}(x d_1), Tr_{q/p}(x d_2), \cdots, Tr_{q/p}(xd_n))\mid x\in\mathbb{F}_q\}.$$ The set $D$ is usually called the {\em defining-set} of the code $\mathcal {C}_D$. The resulting code $\mathcal {C}_D$ is linear over $\mathbb{F}_p$ of length $n$ with  dimension at most $k$. Indeed, the function $d\in \mathbb{F}_q \to Tr_{q/p}(xd)$ matches once with each linear form over $\mathbb{F}_q$ when $x$ ranges over $\mathbb{F}_q$.
\\\\
The idea of a construction of a linear code from a defining set goes back to Wolfmann \cite{Wolfmann} in 1975, who got inspired by the trace construction of irreducible cyclic codes given by L.D. Baumert and R.J. McEliece in 1972 in \cite{Baumert}. However, the first apparition in the literature of the proper second generic construction goes back to 2007 by Ding and Niederreiter in \cite{DingSecondConstruction}. There it is possible to find a first description of the main characteristic of the obtained codes, in particular if, for $x\in\mathbb{F}_q$, we define
$$\textbf{c}_x:=(\text{Tr}(xd_1),\dots,\text{Tr}(xd_n)),$$
then the Hamming weight $w t\left(\mathbf{c}_x\right)$ of $\mathbf{c}_x$ is $n-N_x(0)$, where
$$
N_x(0)=\#\left\{1 \leq i \leq n \mid \operatorname{Tr}_{q / p}\left(x d_i\right)=0\right\}, x \in \mathbb{F}_q
$$
Note that
$$
\begin{aligned}
p N_x(0)= & \sum_{i=1}^n \sum_{y \in \mathbb{F}_p} \xi_p^{y T r_{q / p}\left(x d_i\right)} \\
& =\sum_{i=1}^n \sum_{y \in \mathbb{F}_p} \chi_1\left(y x d_i\right)=n+\sum_{y \in \mathbb{F}_p^*} \chi_1(y x D),
\end{aligned}
$$
where $\chi_1$ is the canonical additive character of $\mathbb{F}_q, a D$ denotes the set $\{a d \mid d \in D\}$ and $\chi_1(S):=\sum_{x \in S} \chi_1(x)$ for any subset $S$ of $\mathbb{F}_q$. Therefore,
$$
w t\left(\mathbf{c}_x\right)=\frac{(p-1)}{p} n-\frac{1}{p} \sum_{y \in \mathbb{F}_p^*} \chi_1(y x D).
$$
In the same paper we have an explicit description of the dimension of the codes, we report it here with its proof.
\begin{thm}[\cite{DingSecondConstruction}, Theorem 6]\label{thm: dimension second construction}
    Let $S$ be the $\mathbb{F}_p$-linear subspace of $\mathbb{F}_q$ spanned by $D$. Then the dimension of the linear code $\mathcal{C}_D$ is equal to the dimension of $S$.
\end{thm}
\begin{proof}
    Consider the $\mathbb{F}_p$-linear map $\phi: x \in \mathbb{F}_q \mapsto \textbf{c}_x \in$ $\mathbb{F}_p^n$. Then
$$
\operatorname{dim}\left(\mathcal{C}_D\right)=m-\operatorname{dim}(\operatorname{ker}(\phi)).
$$
To determine the dimension of the kernel $\operatorname{ker}(\phi)$, we note that $x \in$ $\operatorname{ker}(\phi)$ if and only if $\operatorname{Tr}(x s)=0$ for all $s \in S$. Put $k=\operatorname{dim}(S)$ and let $\left\{b_1, \ldots, b_k\right\}$ be an ordered basis of $S$ over $\operatorname{GF}(q)$. We extend this basis to an ordered basis $\left\{b_1, \ldots, b_m\right\}$ of $\mathbb{F}_q$ over $\mathbb{F}_p$. Let $\left\{e_1, \ldots, e_m\right\}$ be the dual basis of $\left\{b_1, \ldots, b_m\right\}$. We claim that $\operatorname{ker}(\phi)$ is equal to the $\mathbb{F}_p$-linear subspace $V$ of $\mathbb{F}_q$ spanned by $e_{k+1}, \ldots, e_m$. It is obvious from the definition of a dual basis that $V \subseteq \operatorname{ker}(\phi)$. Conversely, if $x \in \operatorname{ker}(\phi)$, then we write
$$
x=x_1 e_1+\cdots+x_m e_m
$$
with $x_j \in \mathbb{F}_p$ for $1 \leq j \leq m$. Since $x \in \operatorname{ker}(\phi)$, it follows that
$$
x_j=\operatorname{Tr}\left(x b_j\right)=0 \text { for } 1 \leq j \leq k
$$
and so $x \in V$. Thus, we have indeed $\operatorname{ker}(\phi)=V$, hence $\operatorname{dim}(\operatorname{ker}(\phi))=\operatorname{dim}(V)=m-k$, which concludes the proof.
\end{proof}

The previous proof gives the parameters of $\mathcal{C}_D$ using a linear algebra approach. We will propose an alternative proof of the theorem (see \ref{lemma: dimension of codes in the second generic construction}). The originality lies in the link between $\mathcal{C}_D$ and $\mathcal{C}_D^\perp$, so the proof follows naturally from the struture of $\mathcal{C}_D^\perp$.
\\\\
The latter formula on the weight of each codeword leaded to a deep study of the weight enumerators of these codes, especially looking for some few weights codes, which have application in authentication codes and secret schemes and are connected to association schemes and graphs (see for example \cite{FewWeights}).
\\\\
In fact, the first result was obtained by Ding and Niederreiter in \cite{DingSecondConstruction} was in the construction of cyclotomic linear codes of order $3$ with few weights, in particular they found two classes of such codes. The first one is the following. 
\begin{thm}[\cite{DingSecondConstruction}, Theorem 7, Theorem 8]
    Let $q \equiv 2(\bmod 3), m$ even, and $r=q^m>4$. Then the first class of linear cyclotomic codes of order 3 $C_D$ over $\mathbb{F}_q$ has two nonzero weights and the parameters
$$
\left[\frac{r-1}{3(q-1)}, m, d\right]
$$
where the minimum distance is given by
$$
d=\left\{\begin{array}{lll}
\frac{r-\sqrt{r}}{3 q}, & \text { if } m \equiv 0 & (\bmod 4) \\
\frac{r-2 \sqrt{r}}{3 q}, & \text { if } m \equiv 2 & (\bmod 4) .
\end{array}\right.
$$
Furthermore, the nonzero weights and their frequencies are
$$\begin{tabular}{|c|c|}
\hline weight & frequency \\
\hline$\frac{r+(-1)^{(m-2) / 2} \sqrt{r}}{3 q}$ & $\frac{2(r-1)}{3}$ \\
\hline$\frac{r-(-1)^{(m-2) / 2} 2 \sqrt{r}}{3 q}$ & $\frac{r-1}{3}$ \\
\hline
\end{tabular}$$
In addition, the dual code has parameters 
$$\left[\frac{r-1}{3(q-1)},\frac{r-1}{3(q-1)}-m, d^\perp\right],$$
with $d^\perp\ge 3$.
\end{thm}
The second class is the following.
\begin{thm}[\cite{DingSecondConstruction}, Theorem 9, Theorem 10]
    Let $q \equiv 2 \quad(\bmod 3), m$ even, and $r=q^m$. Then the second class of linear cyclotomic codes of order $3$ $\mathcal{C}_D$ over $\mathbb{F}_q$ has two nonzero weights and the parameters
$$
\left[\frac{2(r-1)}{3(q-1)}, m, d\right]
$$
where the minimum distance is given by
$$
d=\left\{\begin{array}{lll}
\frac{2 r-2 \sqrt{r}}{3 q}, & \text { if } m \equiv 0 & (\bmod 4) \\
\frac{2 r-\sqrt{r}}{3 q}, & \text { if } m \equiv 2 & (\bmod 4)
\end{array}\right.
$$
Furthermore, the nonzero weights and their frequencies are
$$\begin{tabular}{|c|c|}
\hline weight & frequency \\
\hline$\frac{2 r+(-1)^{(m-2) / 2} 2 \sqrt{r}}{3 q}$ & $\frac{r-1}{3}$ \\
\hline$\frac{2 r-(-1)^{(m-2) / 2} \sqrt{r}}{3 q}$, & $\frac{2(r-1)}{3}$ \\
\hline
\end{tabular}$$
In addition, the dual code has parameters
$$\left[\frac{2(r-1)}{3(q-1)},\frac{2(r-1)}{3(q-1)}-m,d^\perp\right],$$
with $d^\perp\ge 3$.
\end{thm}

\subsection{Some (more) results on linear codes from specific defining sets}
Considering the weight enumerators, in \cite{WeightEnumeratorSecond} it is possible to find the complete weight enumerators of this kind of linear codes in three specific cases:
\begin{enumerate}
    \item When $D$ is a skew Hadamard difference set or Paley type partial difference set in $\mathbb{F}_q$;
    \item When $D=\{f(x)\,:\,x\in\mathbb{F}_q\}\setminus \{0\}$ where $f(x)$ is a quadratic form over $\mathbb{F}_q$;
    \item When $D=\{x\in\mathbb{F}_q\,:\,\text{Tr}_s(x^{p^s+1})=0$\} where $m=2s$ is an even integer.
\end{enumerate}
In the first case the weight enumerator is the following.
\begin{thm}[\cite{WeightEnumeratorSecond}, Theorem 3.2]
     Let $\mathcal{C}_D$ be the second generic construction code.
     \begin{itemize}
         \item Suppose that $D$ is a skew Hadamard difference set in $\mathbb{F}_q$. Then $\mathcal{C}_D$ is $a\left[\frac{p^m-1}{2}, m\right]$ linear code and its complete weight enumerator is
$$
\begin{aligned}
& z_0^{\frac{p^m-1}{2}}+\frac{p^m-1}{2} z_0^{\frac{p^{m-1}-1}{2}} \prod_{\substack{c \in \mathbb{F}_p^* \\
\left(\frac{c}{p}\right)=1}} z_c^{\frac{p^{m-1}+p^{\frac{m-1}{2}}}{2}} \prod_{\substack{c \in \mathbb{F}_p^* \\
\left(\frac{c}{p}\right)=-1}} z_c^{\frac{p^{m-1}-p^{\frac{m-1}{2}}}{2}} \\
& +\frac{p^m-1}{2} z_0^{\frac{p^{m-1}-1}{2}} \prod_{\substack{c \in \mathbb{F}_p^* \\
\left(\frac{c}{p}\right)=1}} z_c^{\frac{p^{m-1}-p^{\frac{m-1}{2}}}{2}} \prod_{\substack{c \in \mathbb{F}_p^* \\
\left(\frac{c}{p}\right)=-1}} z_c^{\frac{p^{m-1}+p^{\frac{m-1}{2}}}{2}} . \\
&
\end{aligned}
$$
         \item Suppose that $D$ is a Paley type partial difference set in $\mathbb{F}_q$. Then $\mathcal{C}_D$ is also $a\left[\frac{p^m-1}{2}, m\right]$ linear code. If $m$ is odd, then the complete weight enumerator of $\mathcal{C}_D$ is given by the latter. If $m$ is even, then the complete weight enumerator of $\mathcal{C}_D$ is
$$
\begin{gathered}
z_0^{\frac{p^m-1}{2}}+\frac{p^m-1}{2} z_0^{\frac{\left(p^{\frac{m}{2}}-1\right)\left(p^{\frac{m}{2}-1}+1\right)}{2}}\left(z_1 z_2 \ldots z_{p-1}\right)^{\frac{p^{m-1}-p^{\frac{m}{2}-1}}{2}} \\
+\frac{p^m-1}{2} z_0^{\frac{\left(p^{\frac{m}{2}}+1\right)\left(p^{\frac{m}{2}-1}-1\right)}{2}}\left(z_1 z_2 \ldots z_{p-1}\right)^{\frac{p^{m-1}+p^{\frac{m}{2}-1}}{2}} .
\end{gathered}
$$
     \end{itemize}
\end{thm}
In the second case the weight enumerator is the following.
\begin{thm}[\cite{WeightEnumeratorSecond}, Theorem 4.2]
    Let $\mathcal{C}_D$ be the second generic construction linear code and $D=\left\{f(x): x \in \mathbb{F}_q\right\} \backslash\{0\}$, where $f(x)$ is a quadratic form over $\mathbb{F}_q$ of rank $r$ such that $f$ is $e$-to-$1$, $f(0)=0$ and $f(x)\not =0$ for all $x\in\mathbb{F}_q^*$. Then 
    \begin{itemize}
        \item If $r$ is odd, then the complete weight enumerator of $\mathcal{C}_D$ is
$$
\begin{aligned}
& z_0^{\frac{p^m-1}{e}}+\frac{p^m-1}{2} z_0^{\frac{p^{m-1}-1}{e}} \prod_{\substack{c \in \mathbb{F}_p^* \\
\left(\frac{c}{p}\right)=1}} z_c^{\frac{p^{m-1}+p^{m-\frac{r+1}{2}}}{e}} \prod_{\substack{c \in \mathbb{F}_p^* \\
\left(\frac{c}{p}\right)=-1}} z_c^{p^{m-1}-p^{m-\frac{r+1}{2}}} \\
& +\frac{p^m-1}{2} z_0^{p^{m-1}-1} \prod_{\substack{c \in \mathbb{F}_p^* \\
\left(\frac{c}{p}\right)=1}} z_c^{p^{m-1}-p^{m-\frac{r+1}{2}}} \prod_{\substack{c \in \mathbb{F}_p^* \\
\left(\frac{c}{p}\right)=-1}} z_c^{p^{m-1}+p^{m-\frac{r+1}{2}}} . \\
&
\end{aligned}
$$
    \item If $r$ is even, then the complete weight enumerator of $\mathcal{C}_D$ is
$$
\begin{aligned}
& z_0^{\frac{p^m-1}{e}}+\frac{p^m-1}{2} z_0^{p^{m-1}+p^{m-\frac{r}{2}}-p^{m-1-\frac{r}{2}-1}} e^e\left(z_1 z_2 \ldots z_{p-1}\right)^{\frac{p^{m-1}-p^{m-1-\frac{r}{2}}}{e}} \\
& \quad+\frac{p^m-1}{2} z_0^{\frac{p^{m-1}-p^{m-\frac{r}{2}}+p^{m-1-\frac{r}{2}-1}}{e}}\left(z_1 z_2 \ldots z_{p-1}\right)^{\frac{p^{m-1}+p^{m-1-\frac{r}{2}}}{e}}.
\end{aligned}
$$
    \end{itemize}
\end{thm}
In the third case the weight enumerator is the following.
\begin{thm}[\cite{WeightEnumeratorSecond}, Theorem 5.1]
    Let $\mathcal{C}_D$ be the second generic construction code and $D=\left\{z \in \mathbb{F}_{p^m}^*: \operatorname{Tr}_s\left(z^{p^s+1}\right)=\right.$ $0\}$, where $m$ is even and $m=2 s$ for an integer $s>1$. Then $\mathcal{C}_D$ is $a\left[\left(p^s+1\right)\left(p^{s-1}-1\right), m\right]$ linear code and its complete weight enumerator is
$$
\begin{aligned}
& z_0^{\left(p^s+1\right)\left(p^{s-1}-1\right)}+\left(p^s+1\right)\left(p^{s-1}-1\right) z_0^{p^{m-2}-p^s+p^{s-1}-1}\left(z_1 z_2 \ldots z_{p-1}\right)^{p^{m-2}} \\
& \quad+\left(p^s+1\right)\left(p^s-p^{s-1}\right) z_0^{p^{m-2}-1}\left(z_1 z_2 \ldots z_{p-1}\right)^{p^{m-2}-p^{s-1}} .
\end{aligned}
$$
\end{thm}
The applications of these codes proposed in the same paper refer mainly to constant composition codes and systhematic authentication codes.
\\\\
Again in the idea of the few weights linear codes, in 2008 a class of two weights codes was presented in \cite{TwoWeightsSecond} as punctured from irreducible cyclic codes and containing a family of optimal codes (including MDS). We find again the use of cyclotomic classes and character theory in finite fields that leaded to the following codes. 
\begin{thm}[\cite{TwoWeightsSecond}, Theorem 2.1]
    Let $p$ be a prime, $q=p^s$, $m=2lk$, $r=q^m$, where $s,l,k$ are positive integers. Let $h$ be a positive divisor of $q^k+1$ and assume $k<\sqrt{r}+1$. Then the obtained code $\mathcal{C}$ is an $[n,m]$ linear code over $\mathbb{F}_q$ and has the following weight distribution
    $$\begin{tabular}{|c|c|}
\hline weight & frequency \\
\hline$\frac{r+(-1)^l(h-1) \sqrt{r}}{q h}$ & $\frac{r-1}{h}$ \\
\hline$\frac{r+(-1)^{l-1} \sqrt{r}}{q h}$ & $\frac{(h-1)(r-1)}{h}$ \\
\hline 0 & 1 \\
\hline
\end{tabular}$$
Futhermore, the dual code $\mathcal{C}^\perp$ of $\mathcal{C}$ is an $[n,n-m,d^\perp]$ linear code with minimum dinstance $d^\perp\ge 3$. 
\end{thm}
In \cite{TwoDesigns} Ding studied some classes of two generic construction codes from specific defining sets. First of all we have the skew sets $D$ (a subset $D\subseteq \mathbb{F}_q$ is called skew if $D$, $-D$ and $\{0\}$ form a partition of $\mathbb{F}_q$). 
\begin{thm}[\cite{TwoDesigns}, Theorem 1]
    Let $D$ be any skew sets of $\mathbb{F}_q$. Then $\mathcal{C}_D$ is a one-weight code over $\mathbb{F}_p$ with parameters $[(q-1)/2,m,(p-1)q/2p]$.
\end{thm}
The idea of the proof lies on the fact that if $D$ is a skew set, then $xD$ is also skew for every $x\in\mathbb{F}_q^*$.
\\\\
Another type of codes presented by Ding in the same paper \cite{TwoDesigns} involves the defining sets as the preimage $f^{-1}(b)$ for some functions $f$ from $\mathbb{F}_{p^m}$ to $\mathbb{F}_p$.\\
In the boolean case, for a function $f$ from $\mathbb{F}_{2^m}$ to $\mathbb{F}_2$, he was able to determine the complete weight enumerator by the use of the Walsh transform. In particular, let $D_f:=f^{-1}(1)$ be the support of $f$ and let $n_f=|D_f|$; then we have the following theorem, with specific corollaries when the function $f$ is bent or semibent. 

\begin{thm}[\cite{TwoDesigns}, Theorem 9]
    If $f$ is not an affine function, then $\mathcal{C}_{D_f}$ is a binary linear code of length $n_f$ and dimension $m$, and its weigth distribution is given by the following multiset 
    $$\big\{\big\{\frac{2n_f+\hat{f}(w)}{4}\,:\,w\in\mathbb{F}_{2^m}^*\big\}\big\}\cup \{\{0\}\}.$$
\end{thm}

\section{The Dual and the Hull code in the first generic construction}

\subsection{The dual code}
Let $f$ be a polynomial from $\mathbb{F}_q$ to $\mathbb{F}_q$, where $q=p^m$. Then, recalling the definition of the dual, we have that 
 $$\mathcal{C}(f)^\perp=\{(c_1,\dots,c_q)\in\mathbb{F}_p^q\,|\,\sum_{i=1}^q c_i Tr_{q/p}(af(x_i)+bx_i)=0\,\text{for every}\, a,b\in\mathbb{F}_q\}=$$
 $$=\{(c_1,\dots,c_q)\in\mathbb{F}_p^q\,|\, Tr_{q/p}\sum_{i=1}^q(c_i(af(x_i)+bx_i))=0\,\text{for every}\, a,b \in\mathbb{F}_q\}$$
thanks to the linearity of the trace map and where we have fixed an order $x_i$ in $\mathbb{F}_q$. Observe that with a different order the two codes would be clearly permutationally equivalent. 
\\\\
In other words, for every $a,b\in\mathbb{F}_q$ we must have 
$$\sum_{i=1}^qc_i(af(x_i)+bx_i)\in\text{ker}(Tr_{q/p})$$
for every $a,b\in\mathbb{F}_q$. Actually, we can have a much more explicit descrpition of the dual code, which is the following.
\begin{prop}\label{prop: dual code in the first generic construction}
    Let $f$ be a polynomial from $\mathbb{F}_q$ to $\mathbb{F}_q$, $\mathcal{C}(f)$ be the code built with the first generic construction, $l_1=(x_1,\dots,x_q)$, $l_2=(f(x_1),\dots,f(x_q))$, $\mathcal{L}_1$ be the code generated by $l_1$ and $\mathcal{L}_2$ the code generated by $l_2$ over $\mathbb{F}_q$. Then 
$$\mathcal{C}(f)^\perp= \mathcal{L}_1^\perp\cap \mathcal{L}_2^\perp\cap \mathbb{F}_p^q.$$
\end{prop}
\begin{proof}
    Consider a codeword $(c_1,\dots, c_q)\in\mathbb{F}_p^q$, then we saw that $(c_1,\dots,c_q)\in\mathcal{C}(f)^\perp$ if and only if
    $$\sum_{i=1}^qc_i\text{Tr}_{q/p}(af(x_i)+bx_i)=\text{Tr}_{q/p}(\sum_{i=1}^qc_i(af(x_i)+bx_i))=0$$
    for every $a,b\in\mathbb{F}_q$. Now, consider the mapping
    $$\begin{aligned}
    \varphi:\mathbb{F}_q^2 &\longrightarrow \mathbb{F}_q\\ (a,b) &\longmapsto \sum_{i=1}^qc_i(af(x_i)+bx_i).
\end{aligned}$$
A straightforward computation shows that $\varphi$ is linear over $\mathbb{F}_q$, hence it is either surjective or the zero mapping. Since  the trace map is not null, our condition is equivalent to 
$$\sum_{i=1}^q c_i(af(x_i)+bx_i)=0$$
for every $a,b\in\mathbb{F}_q$, which in other words means that 
$$\mathcal{C}(f)^\perp=(\mathcal{L}_1+\mathcal{L}_2)^\perp\cap\mathbb{F}_p^q.$$
It is a well known result that the dual of the sum of two linear codes coincides with the intersection of its duals, hence we have our thesis:
$$\mathcal{C}(f)^\perp=(\mathcal{L}_1+\mathcal{L}_2)^\perp\cap\mathbb{F}_p^q=\mathcal{L}_1^\perp\cap\mathcal{L}_2^\perp\cap\mathbb{F}_p^q.$$
\end{proof}
Hence, we have reduced the computation of the dual code to an orthogonality problem for which it is possible to use the Gram-Schmidt algorithm; we recall that in this case the complexity is $\mathcal{O}(q^3)$. 

\subsection{The Walsh transform}
In this paragraph we would like to give a necessary condition for a codeword to be in the dual of a code built with the first generic construction and a specific one when the functions are weakly regular bent.
\\\\
Consider a function $f:\mathbb{F}_q\rightarrow \mathbb{F}_q$, $q=p^m$, which we will use the define the code in the first generic construction, with $p$ odd, and define the following function 
$$\begin{aligned}
    g:\mathbb{F}_q &\longrightarrow \mathbb{F}_p\\ x &\longmapsto \text{Tr}_{q/p}(f(x)-x).
\end{aligned}$$
If we suppose that $g$ is weakly regular bent, then there exists another function $g^*:\mathbb{F}_q\rightarrow\mathbb{F}_p$ such that, for $b\in\mathbb{F}_q$:
$$\hat{\chi}_{g}(b)=\epsilon \sqrt{p^*}^m\zeta^{g^*(b)}$$
and 
$$\hat{\chi}_{g^*}(b)=\frac{\epsilon p^m}{\sqrt{p^*}^m}\zeta^{g(x)},$$
where $\epsilon=\pm 1$ is the sign of the Walsh tranform of $f(x)$, $p^*=(\frac{-1}{p})p$ and $\zeta=e^{\frac{2\pi i}{p}}$ is the primitive $p$-th root of unity. In this case, with the same notation, fixing an enumeration $(x_i)_{i=1,\dots,q}$ in $\mathbb{F}_q$, we have the following necessary condition. 
\begin{prop}\label{prop: first necessary condition in the first generic construction for weakly regular bent functions}
    Let $f:\mathbb{F}_q\rightarrow\mathbb{F}_q$ be a function such that the previously defined function $g$ is weakly regular bent and consider $(c_1,\dots,c_q)\in\mathcal{C}(f)^\perp$. Then
    \begin{enumerate}
        \item if $f$ respects the scalar multiplication, i.e. for every $\alpha\in\mathbb{F}_p$ and $x\in\mathbb{F}_q$, $f(\alpha x)=\alpha f(x)$:
        $$\prod _{i=1}^q \hat{\chi}_{g^*}(c_ix_i)=\big(\frac{p^m}{\epsilon\sqrt{p^*}^m}\big )^q,$$
        or, equivalently,
        $$\mathfrak{Im}(\prod _{i=1}^q \hat{\chi}_{g^*}(c_ix_i))=0;$$
        \item for a generic $f$:
        $$\prod _{i=1}^q \big(\hat{\chi}_{g^*}(x_i)\big)^{c_i}=\prod_{i=1}^q\big(\frac{p^m}{\epsilon\sqrt{p^*}^m}\big )^{c_i},$$
        or, equivalently, 
        $$\mathfrak{Im}(\prod _{i=1}^q \big(\hat{\chi}_{g^*}(x_i)\big)^{c_i})=0.$$
    \end{enumerate}
\end{prop}
\begin{proof}
    If $(c_1,\dots,c_q)\in\mathcal{C}(f)^\perp$ then, from \cref{prop: dual code in the first generic construction} we get that $(c_1,\dots,c_q)\in\mathcal{L}_1^\perp$ and $(c_1,\dots,c_q)\in\mathcal{L}_2^\perp$, which means that 
    $$\sum_{i=1}^q c_ix_i=0\,\,\text{and}\,\,\sum_{i=1}^qc_if(x_i)=0$$
    and so 
    $$\sum_{i=1}^q c_i(f(x_i)-x_i)=0.$$
    Now we distinguish the two cases.
    \begin{enumerate}
        \item If $f$ respects the scalar multiplication, then
        $$\frac{\epsilon\sqrt{p^*}^m}{p^m}\hat{\chi}_{g^*}(c_ix_i)=\zeta^{\text{Tr}_{q/p}(f(c_ix_i)-c_ix_i)}=\zeta^{\text{Tr}_{q/p}(c_i(f(x_i)-x_i))}.$$
        Hence, by considering the product over $i$ from $1$ to $q$, we obtain 
        $$\prod_{i=1}^q\frac{\epsilon\sqrt{p^*}^m}{p^m}\hat{\chi}_{g^*}(c_ix_i)=\zeta^{\text{Tr}_{q/p}(\sum_{i=1}^qc_i(f(x_i)-x_i))}=\zeta^0=1,$$
        which is equivalent to our thesis and also to the fact of the imaginary part of the product being $0$ because of the structure of the Walsh transform with a fixed coefficients multiplied by a $p$-th root of the unity ($p$ is always assumed to be odd).
        \item For a generic $f$ we have 
        $$\big(\frac{\epsilon\sqrt{p^*}^m}{p^m}\hat{\chi}_{g^*}(x_i)\big)^{c_i}=\big(\zeta^{\text{Tr}_{q/p}(f(x_i)-x_i)}\big)^{c_i}=\zeta^{\text{Tr}_{q/p}(c_i(f(x_i)-x_i))}.$$
        Hence, by considering the product over $i$ from $1$ to $q$, we obtain 
        $$\prod_{i=1}^q\big(\frac{\epsilon\sqrt{p^*}^m}{p^m}\hat{\chi}_{g^*}(x_i)\big)^{c_i}=\zeta^{\text{Tr}_{q/p}(\sum_{i=1}^qc_i(f(x_i)-x_i))}=\zeta^0=1,$$
        which is equivalent to our thesis and also to the one involving the imaginary part for the same reasons explained in point $1$.
    \end{enumerate}
\end{proof}
\begin{remark}
    Observe that if $f$ is linear the previous expression would simplify, however we can not suppose that in order to have a weakly regular bent function $g$. In fact, if $f$ is linear, then clearly $g$ would be linear and so
    $$g(x)=\text{Tr}_{q/p}(ax)$$
    for a suitable $a\in\mathbb{F}_q$. For this reason, computing the Walsh-Hadamard coefficients, we would get, for $b\in\mathbb{F}_q$:
    $$\hat{\chi}_g(b)=\sum_{x\in\mathbb{F}_q} \zeta^{\text{Tr}_{q/p}(ab)-\text{Tr}_{q/p}(bx)}=\sum_{x\in\mathbb{F}_q}\zeta^{\text{Tr}_{q/p}}(b(a-x)),$$
    which is equal to $0$ if $b\ne 0$ and to $q$ if $b=0$, for a well known result in character theory. Hence, clearly $g$ would not be weakly regular bent in this case.
\end{remark}
In a similar way, we can give also another necessary condition for a codeword to be in the dual code by considering the function 
$$\begin{aligned}
    g:\mathbb{F}_q &\longrightarrow \mathbb{F}_p\\ x &\longmapsto \text{Tr}_{q/p}(f(x)),
\end{aligned}$$
if it is as well weakly regular bent.
\begin{prop} \label{prop: second necessary condition in the first generic construction for weakly regular bent functions}
    Let $f:\mathbb{F}_q\rightarrow\mathbb{F}_q$ be a function such that the function $g=\text{Tr}_{q/p}(f)$ is weakly regular bent and consider $(c_1,\dots,c_q)\in\mathcal{C}(f)^\perp$. Then
    \begin{enumerate}
        \item if $f$ respects the scalar multiplication:
        $$\prod _{i=1}^q \hat{\chi}_{g^*}(c_ix_i)=\big(\frac{p^m}{\epsilon\sqrt{p^*}^m}\big )^q,$$
        or, equivalently, for an odd $p$,
        $$\mathfrak{Im}(\prod _{i=1}^q \hat{\chi}_{g^*}(c_ix_i))=0;$$
        \item for a generic $f$:
        $$\prod _{i=1}^q \big(\hat{\chi}_{g^*}(x_i)\big)^{c_i}=\prod_{i=1}^q\big(\frac{p^m}{\epsilon\sqrt{p^*}^m}\big )^{c_i},$$
        or, equivalently, for an odd $p$,
        $$\mathfrak{Im}(\prod _{i=1}^q \big(\hat{\chi}_{g^*}(x_i)\big)^{c_i})=0.$$
    \end{enumerate}
\end{prop}
\begin{proof}
    The proof is analogous to the one of \ref{prop: first necessary condition in the first generic construction for weakly regular bent functions}, if $(c_1,\dots,c_q)\in\mathcal{C}(f)^\perp$ then, from \cref{prop: dual code in the first generic construction} we get that $(c_1,\dots,c_q)\in\mathcal{L}_2^\perp$, which means that 
    $$\sum_{i=1}^qc_if(x_i)=0$$
    Now we distinguish the two cases.
    \begin{enumerate}
        \item If $f$ respects the scalar multiplication, then
        $$\frac{\epsilon\sqrt{p^*}^m}{p^m}\hat{\chi}_{g^*}(c_ix_i)=\zeta^{\text{Tr}_{q/p}(f(c_ix_i))}=\zeta^{\text{Tr}_{q/p}(c_if(x_i))}.$$
        Hence, by considering the product over $i$ from $1$ to $q$, we obtain 
        $$\prod_{i=1}^q\frac{\epsilon\sqrt{p^*}^m}{p^m}\hat{\chi}_{g^*}(c_ix_i)=\zeta^{\text{Tr}_{q/p}(\sum_{i=1}^qc_if(x_i))}=\zeta^0=1,$$
        which is equivalent to our thesis and also to the fact of the imaginary part of the product being $0$ again because of the structure of the Walsh transform with a fixed coefficients multiplied by a $p$-th root of the unity (again, $p$ is always assumed to be odd).
        \item For a generic $f$ we have 
        $$\big(\frac{\epsilon\sqrt{p^*}^m}{p^m}\hat{\chi}_{g^*}(x_i)\big)^{c_i}=\big(\zeta^{\text{Tr}_{q/p}(f(x_i))}\big)^{c_i}=\zeta^{\text{Tr}_{q/p}(c_if(x_i)
        )}.$$
        Hence, by considering the product over $i$ from $1$ to $q$, we obtain 
        $$\prod_{i=1}^q\big(\frac{\epsilon\sqrt{p^*}^m}{p^m}\hat{\chi}_{g^*}(x_i)\big)^{c_i}=\zeta^{\text{Tr}_{q/p}(\sum_{i=1}^qc_if(x_i))}=\zeta^0=1,$$
        which is equivalent to our thesis and also to the one involving the imaginary part for the same reasons explained in point $1$.
    \end{enumerate}
\end{proof}
\begin{remark}
    Observe that we can not consider $f$ to be linear or in a similar way the function $g=\text{Tr}_{q/p}$ in order to utilise the other equality 
    $$\sum_{i=1}^qc_ix_i=0,$$
    because in this case $g$ would be linear, and so it could not be weakly regular bent, as we observed in the previous remark.
\end{remark}
After considering a weakly regular bent function, now we would like to establish a connection between the dual code in the first generic construction and the Walsh transform in the general case. 
\\\\
Consider a function $f:\mathbb{F}_q\rightarrow\mathbb{F}_q$, $q=p^m$, which we will use to define the code in the first generic construction, and for every $i=1,\dots,q$ define the functions $g_i:\mathbb{F}_q\rightarrow\mathbb{F}_p$ such that $g_i(x_i)=\text{Tr}_{q/p}(f(x_i))$ and $g_i(x)=\text{Tr}_{q/p}(x)$ for $x\ne x_i$. Then we have the following necessary condition.
\begin{prop}\label{prop: first necessary condition in the first general construction for the general case}
    Let $f:\mathbb{F}_q\rightarrow\mathbb{F}_q$ be a function, $g_i:\mathbb{F}_q\rightarrow\mathbb{F}_p$ be the previously defined functions and consider $(c_1,\dots,c_q)\in\mathcal{C}(f)^\perp$. Then 
    $$\prod_{i=1}^q\big(\hat{\chi}_{g_i}(1)+1-q\big)^{c_i}=1.$$
\end{prop}
\begin{proof}
    As we observed in the proof of \ref{prop: first necessary condition in the first generic construction for weakly regular bent functions}, from \ref{prop: dual code in the first generic construction} we have that 
$$\sum_{i=1}^q c_i(f(x_i)-x_i)=0.$$
If we compute the first Walsh-Hadamard coefficient of the function $g_i$ we get the following result:
$$\hat{\chi}_{g_i}(1)=\sum_{x\in\mathbb{F}_q}\zeta^{g_i(x)-\text{Tr}_{q/p}(x)}=q-1+\zeta^{\text{Tr}_{q/p}(f(x_i)-x_i)},$$
and so 
$$\big(\hat{\chi}_{g_i}(1)+1-q\big)^{c_i}=\zeta^{\text{Tr}_{q/p}(c_i(f(x_i)-x_i))}.$$
Hence we obtain our thesis:
$$\prod_{i=1}^q\big(\hat{\chi}_{g_i}(1)+1-q\big)^{c_i}=\zeta^{\text{Tr}_{q/p}(\sum_{i=1}^qc_i(f(x_i)-x_i))}=\zeta^0=1.$$
\end{proof}
In a similar way we can get two other necessary conditions. 
\begin{prop}\label{prop: second necessary condition in the first general construction for the general case}
    Let $f:\mathbb{F}_q\rightarrow\mathbb{F}_q$ be a function, $g_i:\mathbb{F}_q\rightarrow\mathbb{F}_p$ for $i=1,\dots,q$ such that $g_i(x_i)=\text{Tr}_{q/p}(f(x_i)+x_i)$ and $g_i(x)=\text{Tr}_{q/p}(x)$ for $x\ne x_i$ and consider $(c_1,\dots,c_q)\in\mathcal{C}(f)^\perp$. Then 
    $$\prod_{i=1}^q\big(\hat{\chi}_{g_i}(1)+1-q\big)^{c_i}=1.$$
\end{prop}
\begin{proof}
    From \ref{prop: dual code in the first generic construction} we have that 
$$\sum_{i=1}^q c_if(x_i)=0.$$
If we compute the first Walsh-Hadamard coefficient of the function $g_i$ we get the following result:
$$\hat{\chi}_{g_i}(1)=\sum_{x\in\mathbb{F}_q}\zeta^{g_i(x)-\text{Tr}_{q/p}(x)}=q-1+\zeta^{\text{Tr}_{q/p}(f(x_i) },$$
and so 
$$\big(\hat{\chi}_{g_i}(1)+1-q\big)^{c_i}=\zeta^{\text{Tr}_{q/p}(c_if(x_i))}.$$
Hence we obtain our thesis:
$$\prod_{i=1}^q\big(\hat{\chi}_{g_i}(1)+1-q\big)^{c_i}=\zeta^{\text{Tr}_{q/p}(\sum_{i=1}^qc_if(x_i))}=\zeta^0=1.$$
\end{proof}
If we consider other functions $g_i$ we obtain the following.
\begin{prop}\label{prop: third necessary condition in the first general construction for the general case}
    Let $f:\mathbb{F}_q\rightarrow\mathbb{F}_q$ be a function, $g_i:\mathbb{F}_q\rightarrow\mathbb{F}_p$ for $i=1,\dots,q$ such that $g_i(x_i)=\text{Tr}_{q/p}(2x_i)$ and $g_i(x)=\text{Tr}_{q/p}(x)$ for $x\ne x_i$ and consider $(c_1,\dots,c_q)\in\mathcal{C}(f)^\perp$. Then 
    $$\prod_{i=1}^q\big(\hat{\chi}_{g_i}(1)+1-q\big)^{c_i}=1.$$
\end{prop}
\begin{proof}
    From \ref{prop: dual code in the first generic construction} we have that 
$$\sum_{i=1}^q c_ix_i=0.$$
If we compute the first Walsh-Hadamard coefficient of the function $g_i$ we get the following result:
$$\hat{\chi}_{g_i}(1)=\sum_{x\in\mathbb{F}_q}\zeta^{g_i(x)-\text{Tr}_{q/p}(x)}=q-1+\zeta^{\text{Tr}_{q/p}(x_i)},$$
and so 
$$\big(\hat{\chi}_{g_i}(1)+1-q\big)^{c_i}=\zeta^{\text{Tr}_{q/p}(c_ix_i)}.$$
Hence we obtain our thesis:
$$\prod_{i=1}^q\big(\hat{\chi}_{g_i}(1)+1-q\big)^{c_i}=\zeta^{\text{Tr}_{q/p}(\sum_{i=1}^qc_ix_i)}=\zeta^0=1.$$
\end{proof}
From these necessary conditions we can obtain the following restriction on the computation of the dual and also a precise computation when the dimension of the code is equal to one, using character theory.
\begin{prop}\label{prop: dual of the first generic construction and character theory}
    Consider the functions $g_i:\mathbb{F}_q\rightarrow\mathbb{F}_p$ of $\ref{prop: first necessary condition in the first general construction for the general case}$, \ref{prop: second necessary condition in the first general construction for the general case} or \ref{prop: third necessary condition in the first general construction for the general case}, then the function
    $$\begin{aligned}
    \varphi:\mathbb{F}_p^q &\longrightarrow \mathbb{C}\\ (c_1,\dots,c_q) &\longmapsto \prod_{i=1}^q\big(\hat{\chi}_{g_i}(1)+1-q\big)^{c_i}.
\end{aligned}$$
is an additive character of $\mathbb{F}_{p^q}$ after the identification of $\mathbb{F}_p^q$ with $\mathbb{F}_{p^q}$ such that 
$$\mathcal{C}(f)^\perp\subseteq\text{ker}(\varphi).$$
In particular, if $\text{dim}(\mathcal{C}(f))=1$, then
$$\mathcal{C}(f)^\perp=\text{ker}(\varphi).$$
\end{prop}
\begin{proof}
    The fact that $\varphi$ is an additive character of $\mathbb{F}_{p^q}$ after the identification of the two linear spaces is a straightforward computation, then the inclusion 
    $$\mathcal{C}(f)^\perp\subseteq\text{ker}(\varphi)$$
    derives from the necessary conditions expressed in $\ref{prop: first necessary condition in the first general construction for the general case}$, \ref{prop: second necessary condition in the first general construction for the general case} and \ref{prop: third necessary condition in the first general construction for the general case}.
    \\\\
    Thanks to character theory we can say more, in fact we can identify $\varphi$ with an additive character 
    $$\begin{aligned}
    \psi:\mathbb{F}_{p^q} &\longrightarrow \mathbb{C}\\ x &\longmapsto \zeta^{\text{Tr}_{p^q/p}(ax)},
\end{aligned}$$
for a suitable $a\in\mathbb{F}_{p^q}$. Then 
$$\text{ker}(\psi)=a^{-1}\text{ker}(\text{Tr}_{p^q/p}),$$
which is a linear subspace of dimension $q-1$ over $\mathbb{F}_p$. For this reason, if the code $\mathcal{C}(f)$ is $1$-dimensional, then the inclusion is actually an equality.
\end{proof}
\begin{remark}
    The proof of \ref{prop: dual of the first generic construction and character theory} gives us an explicit construction of the kernel of $\varphi$, which can help in the construction of the dual code, especially because we already know the structure of the kernel of the trace map (its element are all of the shape $\alpha^p-\alpha$ for a suitable $\alpha\in\mathbb{F}_{p^q}$, see for instance \cite{lidl1994introduction}).
\end{remark}

\subsection{A practical computation}
If we apply the necessary conditions presented in the previous two sections to actually compute the dual codes in the two generic constructions, what we find out is that, as a bound, we get a linear code of codimension one in which the dual code is contained. 
\\\\
For some practical examples, consider the following table of the known univariate weakly regular bent functions over $\mathbb{F}_{p^m}$, $p$ odd, presented in \cite{WeaklyFunct}:
$$\begin{tabular}{|c|c|c|c|}
\hline $\begin{array}{c}\text { (Weakly regular) bent } \\
\text { function }\end{array}$ & $m$ & $p$ & Reference \\
\hline$\sum_{i=0}^{[m / 2]} T r_1^m\left(c_1 x^{p^x+1}\right)$ & any & any & $\begin{array}{c}{\cite{28, 34, 37}} \\
\text { etc. }\end{array}$ \\
\hline $\begin{array}{c}\sum_{i=0}^{p^k-1} \operatorname{Tr}_1^m\left(c_i x^{i\left(p^k-1\right)}\right)+ \\
\operatorname{Tr}_1^{l}\left(\epsilon x^{\frac{p^m-1}{\epsilon}}\right)\end{array}$ & $m=2 k$ & any & $\begin{array}{c}{\cite{28,33,36}} \\
\text { etc. }\end{array}$ \\
\hline $\operatorname{Tr}_1^m\left(c x^{\frac{3 m-1}{4}}+3^k+1\right)$ & $m=2 k$ & $p=3$ & \cite{27} \\
\hline $\operatorname{Tr}_1^m\left(x^{p^{3 k}+p^{2 k}-p^k+1}+x^2\right)$ & $m=4 k$ & any & {\cite{29}} \\
\hline $\begin{array}{l}\operatorname{Tr}_1^m\left(c x^{\frac{3^i+1}{2}}\right.), i \text { odd } \\
\operatorname{gcd}(i, m)=1\end{array}$ & any & $p=3$ & {\cite{18}} \\
\hline
\end{tabular}$$
For example, we can consider the last function of the table with $p=3, m=2, c=1$ and $i=3$; then we get the function
$$\text{Tr}_1^2(x^6).$$
We could use \ref{prop: first necessary condition in the first generic construction for weakly regular bent functions} in order to study the dual code of the linear code obtained via the first generic construction with the function
$$f(x)=x^6.$$
With the software MAGMA we conclude that the dual code is contained in the linear code described by the following parity check equation
$$X_2+2X_4+X_6+2X_8=0,$$
which represents the necessary condition given by the theorem.

\subsection{The Hull code}

Thanks to the study of the dual code in the first generic construction, we are able to characterize the Hull code in this case in the following explicit way. 

\begin{prop}\label{prop: hull code in the first generic construction}
    Let $f$ be a polynomial from $\mathbb{F}_q$ to $\mathbb{F}_q$, $\mathcal{C}(f)$ be the code built with the first generic construction, $l_1=(x_1,\dots,x_q)$, $l_2=(f(x_1),\dots,f(x_q))$, $\mathcal{L}_1$ be the code generated by $l_1$ and $\mathcal{L}_2$ the code generated by $l_2$ over $\mathbb{F}_q$. Then 
$$\text{Hull}(C(f))= \mathcal{L}_1^\perp\cap \mathcal{L}_2^\perp\cap C(f).$$
\end{prop}
\begin{proof}
    The assertion comes directly from theorem \ref{prop: dual code in the first generic construction}, in fact 
    $$\text{Hull}(C(f))=C(f)^\perp\cap C(f)= \mathcal{L}_1^\perp\cap \mathcal{L}_2^\perp\cap\mathbb{F}_p^q\cap C(f)=\mathcal{L}_1^\perp\cap \mathcal{L}_2^\perp\cap C(f).$$
\end{proof}
For this reason, we can characterize the Hull code also in the following way.
\begin{prop}\label{prop: hull code in the first generic construction via kernel}
    Let $f$ be à function from $\mathbb{F}_q$ to $\mathbb{F}_q$ and $\mathcal{C}(f)$ be the code built with the first generic construction. Define 
    $$\textbf{c}_{\alpha,\beta}:=(\text{Tr}_{q/p}(\alpha f(x)+\beta x))_{x\in\mathbb{F}_q}$$
    and consider the following linear mapping
    $$\begin{aligned}
    \varphi:C(f) &\longrightarrow \mathbb{F}_q^2\\ \textbf{c}_{\alpha,\beta} &\longmapsto (\sum_{i=1}^q \text{Tr}_{q/p}(\alpha x_i+\beta x_i)x_i,\sum_{i=1}^q \text{Tr}_{q/p}(\alpha x_i+\beta x_i)f(x_i)).
\end{aligned}$$
Then 
$$\text{Hull}(C(f))=\text{ker}(\varphi)$$
and in particular 
$$\text{dim}(\text{Hull}(C(f)))=l\,\,\iff \,\, \text{rk}(\varphi)=\text{dim}(C(f))-l.$$
\end{prop}
\begin{proof}
    A quick computation shows that the mapping is actually linear and from \ref{prop: hull code in the first generic construction}, we have that $\textbf{c}_{\alpha,\beta}\in\text{Hull}(C(f))$ if and only if
    $$\begin{cases}
    
       \sum_{i=1}^q \text{Tr}_{q/p}(\alpha x_i+\beta x_i)x_i=0 \\
  \sum_{i=1}^q \text{Tr}_{q/p}(\alpha x_i+\beta x_i)f(x_i)=0 
\end{cases},$$
so this let us conclude that $\text{Hull}(C(f))=\text{ker}(\varphi)$ and also the assertion on the dimension of the Hull code. Notice that when the dimension of $C(f)$ is maximal, hence it is equal to $2m$ (which happens often, for details see \cite{WeightWalshTransform}), then 
$$\text{dim}(\text{Hull}(C(f)))=l\,\,\iff \,\, \text{rk}(\varphi)=2m-l.$$
\end{proof}
Considering that a codeword in the Hull code is actually in the dual as well, we can restate all the proposition in section 4 to obtain a necessary condition for a codeword in $C(f)$ to be in the Hull; as an example, we propose the first proposition, the others can be obtained in the same way. 
\\\\
Consider a function $f:\mathbb{F}_q\rightarrow \mathbb{F}_q$, $q=p^m$, which we will use the define the code in the first generic construction, with $p$ odd, and define the following function 
$$\begin{aligned}
    g:\mathbb{F}_q &\longrightarrow \mathbb{F}_p\\ x &\longmapsto \text{Tr}_{q/p}(f(x)-x).
\end{aligned}$$
If we suppose that $g$ is weakly regular bent, then there exists another function $g^*:\mathbb{F}_q\rightarrow\mathbb{F}_p$ such that, for $b\in\mathbb{F}_q$:
$$\hat{\chi}_{g}(b)=\epsilon \sqrt{p^*}^m\zeta^{g^*(b)}$$
and 
$$\hat{\chi}_{g^*}(b)=\frac{\epsilon p^m}{\sqrt{p^*}^m}\zeta^{g(x)},$$
where $\epsilon=\pm 1$ is the sign of the Walsh transform of $f(x)$, $p^*=(\frac{-1}{p})p$ and $\zeta=e^{\frac{2\pi i}{p}}$ is the primitive $p$-th root of unity. In this case, with the same notation, fixing an enumeration $(x_i)_{i=1,\dots,q}$ in $\mathbb{F}_q$, we have the following necessary condition. 
\begin{prop}\label{prop: hull first necessary condition in the first generic construction for weakly regular bent functions}
    Let $f:\mathbb{F}_q\rightarrow\mathbb{F}_q$ be a function such that the previously defined function $g$ is weakly regular bent and consider $\textbf{c}_{\alpha,\beta}\in\mathcal{C}(f)$. If $\textbf{c}_{\alpha,\beta}\in\text{Hull}(C(f))$ then
    \begin{enumerate}
        \item if $f$ respects the scalar multiplication:
        $$\prod _{i=1}^q \hat{\chi}_{g^*}(\text{Tr}_{q/p}(\alpha f(x_i)+\beta x_i)x_i)=\big(\frac{p^m}{\epsilon\sqrt{p^*}^m}\big )^q,$$
        or, equivalently,
        $$\mathfrak{Im}(\prod _{i=1}^q \hat{\chi}_{g^*}(\text{Tr}_{q/p}(\alpha f(x_i)+\beta x_i)x_i))=0;$$
        \item for a generic $f$:
        $$\prod _{i=1}^q \big(\hat{\chi}_{g^*}(x_i)\big)^{\text{Tr}_{q/p}(\alpha f(x_i)+\beta x_i)}=\prod_{i=1}^q\big(\frac{p^m}{\epsilon\sqrt{p^*}^m}\big )^{\text{Tr}_{q/p}(\alpha f(x_i)+\beta x_i)},$$
        or, equivalently, 
        $$\mathfrak{Im}(\prod _{i=1}^q \big(\hat{\chi}_{g^*}(x_i)\big)^{\text{Tr}_{q/p}(\alpha f(x_i)+\beta x_i)})=0.$$
    \end{enumerate}
\end{prop}

Also, in a similar way as with the dual, we can obtain the following restriction on the computation of the Hull.
\begin{prop}\label{prop: hull of the first generic construction and character theory}
    Consider the functions $g_i:\mathbb{F}_q\rightarrow\mathbb{F}_p$ of $\ref{prop: first necessary condition in the first general construction for the general case}$, \ref{prop: second necessary condition in the first general construction for the general case} or \ref{prop: third necessary condition in the first general construction for the general case}, then the function
    $$\begin{aligned}
    \varphi:\mathcal{C}(f) &\longrightarrow \mathbb{C}\\ \textbf{c}_{\alpha,\beta} &\longmapsto \prod_{i=1}^q\big(\hat{\chi}_{g_i}(1)+1-q\big)^{\text{Tr}_{q/p}(\alpha f(x_i)+\beta x_i)}.
\end{aligned}$$
is an additive character of $\mathcal{C}(f)$ after the usual identifications such that 
$$\text{Hull}(\mathcal{C}(f))\subseteq\text{ker}(\varphi).$$
\end{prop}

\section{The Dual and the Hull code in the second generic construction}

\subsection{The dual code}
We fix a set $D=\{d_1,\dots,d_n\}$ in $\mathbb{F}_q$ (where $q=p^k$). Then, recalling the definition of the dual, we have that 
 $$\mathcal{C}_D^\perp=\{(c_1,\dots,c_n)\in\mathbb{F}_p^n\,|\,\sum_{i=1}^n c_i Tr_{q/p}(xd_i)=0\,\text{for every}\, x\in\mathbb{F}_q\}=$$
 $$=\{(c_1,\dots,c_n)\in\mathbb{F}_p^n\,|\, Tr_{q/p}(x\sum_{i=1}^nc_id_i)=0\,\text{for every}\, x \in\mathbb{F}_q\}$$
thanks to the linearity of the trace map. This gives that 
$$\sum_{i=1}^nc_id_i=0,$$
since the trace map is not the zero mapping. Hence we have just proven the following characterization of the dual code in the second generic construction, particularly straightforward and useful to deal with.
\begin{prop}\label{prop: the dual code in the second generic construction}
Let $D=\{d_1,\dots,d_n\}\subseteq\mathbb{F}_q$ be a defining set for a code in the second generic construction, and let $\mathcal{L}$ be the linear code over $\mathbb{F}_q$ generated by the codeword $l=(d_1,\dots,d_n)$. Then
$$\mathcal{C}_D^\perp=\mathcal{L}^\perp\cap\mathbb{F}_p^n.$$
\end{prop}
If we do not want to refer to the codewords with entrance in $\mathbb{F}_p$, we can observe that $\sum_{i=1}^nc_id_i=0$ with $c_i\in\mathbb{F}_p$, if and only if 
$$\sum_{i=1}^n c_id_i^{p^j}=0$$
for every $j=0,\dots,m-1$. Hence if we define $\mathcal{L}_j$ as the linear code over $\mathbb{F}_q$ generated by the codeword $l_j=(d_1^{p^j},\dots,d_n^{p^j})$, then 
$$\mathcal{C}_D^\perp= \mathcal{L}_j^\perp\cap \mathbb{F}_p,$$
for every $j=0,\dots,m-1$, which can be a useful way to estimate or effectively compute the dual code.
\\\\
Also in this case, we reduce the computation of the dual code to an orthogonality problem; here the Gram-Schmidt algorithm has a complexity of $\mathcal{O}(n^3)$.
\\\\
\begin{remark}
The latter proposition gives us information also on the minimum distance of the dual code, which is strictly related to the defining set $D$, in particular to the $\emph{minimum strictly dependent subsets}$ of $D$, where a subset $S\subseteq D$ is minimum strictly dependent in $D$ if there exists a linear combination 
$$\sum_{d\in S}\alpha_d d=0,$$
with $\alpha_d\in\mathbb{F}_p^\times$ for every $d\in S$ and such that its cardinality is the minimum possible to have this property. 
\\\\
For example, if we assume that the defining set does not contain a zero element, then the dimension $d^*$ of the dual code $\mathcal{C}_D$ is greater or equal to two. Also:
\begin{itemize}
    \item $d^*\ge 3$ if and only if in $D$ there are no two elements multiple of each others;
    \item $d^*\ge 4$ if and only if in $D$ there are no three elements such that there exists a linear combination (not necessary non zero) which equals zero. And so on.
\end{itemize}
It is interesting to notice that for linear code over $\mathbb{F}_2$ this study refers to the subsets of $D$ whose elements sum to zero.
\\\\
With this principle, we can determine some MDS codes using the fact that the dual of an MDS code is again MDS \cite{DualMDS}. By definition, in an MDS code the Singleton bound is reached with equality, hence
$$d=n-k+1.$$
If we consider the dual code, then 
$$d=n-(n-k)+1=k+1.$$
To construct such a code with the use of the second generic construction, we have to find a defining set $D$ such that every $k$ elements are linearly independent. 
\begin{itemize}
    \item For $n=k$ we obtain the trivial case $\mathcal{C}_D=\mathbb{F}_p^k$;
    \item For $n=k+1$ we obtain a $[k+1,k,2]$ MDS linear code whose dual is an $[k+1,1,k+1]$ MDS linear code. In order to do this, it suffices to choose $d_1,\dots,d_k$ linear independent and 
    $$d_{k+1}=\sum_{i=1}^k \alpha_id_i$$
    with $\alpha_i\in\mathbb{F}_p^\times$ for every $i=1,\dots,k$;
    \item For $n=k+2$ we obtain a $[k+2,k,3]$ MDS linear code whose dual is an $[k+2,2,k+1]$ MDS linear code. In order to do this, it suffices to choose $d_1,\dots,d_k$ linear independent and 
    $$d_{k+1}=\sum_{i=1}^k \alpha_id_i;$$
    $$d_{k+2}=\sum_{i=1}^kd_i,$$
    with $\alpha_i\in\mathbb{F}_p^\times$ and $\alpha_i\ne \alpha_j$ for $i\ne j$. 
\end{itemize}
\end{remark}

It is interesting that the study of the dual code, and in particular the \ref{prop: the dual code in the second generic construction}, can give us information about the dimension of the generated codes. 
\begin{remark}
Notice that if $D=\{d_1,\dots,d_n\}\subseteq\mathbb{F}_q$ is a defining set for a code in the second generic construction such that its element are linearly independent over $\mathbb{F}_p$, then
$$\mathcal{C}_D=\mathbb{F}_p^n.$$
In fact, thanks to \ref{prop: the dual code in the second generic construction}, if we take $\mathcal{L}$ to be the linear code generated over $\mathbb{F}_q$ by the codeword $l=(d_1,\dots,d_n)$, then 
$$\mathcal{C}_D^\perp=\mathcal{L}\cap\mathbb{F}_p^n=\{0\},$$
since, by definition of linear independence, the only null linear combination is the trivial one. This is clearly equivalent to the fact that $\mathcal{C}_D=\mathbb{F}_p^n$.
\end{remark}
It is interesting how we can generalize this idea to compute the exact dimension of the code in the second generic construction.
\begin{lemma}\label{lemma: dimension of codes in the second generic construction}
    Let $D=\{d_1,\dots,d_n\}\subseteq \mathbb{F}_q$ be a defining set for a code in the second generic construction, then
    $$\text{dim}(\mathcal{C}_D)=\text{dim}_{\mathbb{F}_p}\langle d_1,\dots,d_n\rangle.$$
\end{lemma}
\begin{proof}
    After having chosen a basis of $\mathbb{F}_q$ over $\mathbb{F}_p$, we write the elements $d_i$ as vectors in this basis, then 
$$M:=(d_1,\dots,d_n)$$
is a $m\times n$ matrix with entrances in $\mathbb{F}_p$ such that 
$$\mathcal{C}_D^\perp=\text{ker}(M),$$
so $M$ is \emph{almost} a parity check matrix for the dual code, where the adverb \emph{almost} is necessary because the matrix does not have in general the right rank of a parity check one.
\\\\
Hence, we have our dimensions:
$$\text{dim}(\mathcal{C}_D^\perp)=n-\text{dim}_{\mathbb{F}_p}\langle d_1,\dots, d_n\rangle.$$
And finally 
$$\text{dim}(\mathcal{C}_D)=\text{dim}_{\mathbb{F}_p}\langle d_1,\dots, d_n\rangle.$$
\end{proof}
\begin{remark} \label{remark: generator matrix second generic construction}
    If we suppose that $\text{dim}_{\mathbb{F}_p}\langle d_1,\dots,d_n\rangle=k$ and that $d_1,\dots,d_k$ are the $k$ linearly independent elements, then 
    $$d_{k+1}=\sum_{i=1}^k \alpha_i^1 d_i,\dots, d_n=\sum_{i=1}^k\alpha_i^{n-k}d_i,$$
    for suitable $\alpha_i^j\in\mathbb{F}_p$. Then, from \ref{prop: dual code in the first generic construction}, we have that a codeword $(c_1,\dots,c_n)$ is in the dual of $\mathcal{C}_D$ if and only if
    $$
    \begin{cases}
        c_1+c_{k+1}\alpha_1^1+\dots+c_n\alpha_1^{n-k}=0\\
        c_2+c_{k+1}\alpha_2^1+\dots+c_n\alpha_2^{n-k}=0\\
        \dots\\
        c_k+c_{k+1}\alpha_k^1+\dots+c_n\alpha_k^{n-k}=0
    \end{cases}
    $$
    Hence, if $M:=(d_1\,d_2\,\dots\,d_n)$ is the $k\times n$ matrix where we express by columns the elements of the defining set in coordinates with respects to the basis $\{d_1,\dots,d_k\}$, then $M$, having the following structure with $P$ a $k\times n-k$ matrix,  
    $$M=(I_k\,|\,P)$$
    is a parity check matrix for $\mathcal{C}_D^\perp$, hence a generating matrix for $\mathcal{C}_D$, already in standard form. 
\end{remark}
\begin{remark}
As we observed before, the dimension of a code built with the second generic construction is less or equal to $m$, then we get the following lower bound on the dimension of the dual code:
$$\text{dim}(\mathcal{C}_D^\perp)\geq n-m. $$
One might ask the following reasonable question: is the dual of a code built with the second generic construction again a code that can be made explicit with the second generic construction? The answer in general is no, we can prove it by considering the dimensions of the codes. In fact, if we take a defining set $D$ composed by $n$ elements such that $n> 2m$, then 
$$\text{dim}(\mathcal{C}_D^\perp)\geq n-m > 2m -m=m.$$
This proves that $\mathcal{C}_D^\perp$ can not be made explicit using the second generic construction because if this was the case, then we would have 
$$\text{dim}(\mathcal{C}_D^\perp)\leq m,$$
as it happens for every second generic construction code. 
\end{remark}
One may also ask if a general linear code can be made explicit using the second generic construction over some finite fields $\mathbb{F}_{p^m}$; the answer is yes and can be found in the following proposition, which stats that every linear code is a second generic construction code for a sufficiently large value of $m$. We propose an alternative proof compared to the one in \cite{ref1} and \cite{ref2}.
\begin{prop}\label{LinearCodeToSecondConstr}
    Let $\mathcal{C}$ be a $[n,k,d]$ linear code over $\mathbb{F}_p$ and let $m\in\mathbb{N}$. Then there exists a defining set $D\in\mathbb{F}_q^n$ such that 
    $$\mathcal{C}=\mathcal{C}_D$$
    if and only if $m\ge k$. 
\end{prop}
\begin{proof}
    Let $q=p^m$, $D=\{d_1,\dots,d_n\}\in\mathbb{F}_q^n$ and consider the mapping 
    $$\begin{aligned}
    \varphi:\mathbb{F}_{q}^n &\longrightarrow \text{Hom}_{\mathbb{F}_p}(\mathbb{F}_q,\mathbb{F}_p^n)\\ D &\longmapsto \varphi_D,
\end{aligned}$$
where $\varphi_D(x):=\textbf{c}_x=(\text{Tr}_{q/p}(xd_i))_{i=1,\dots,n}$, a codeword of the linear code $\mathcal{C}_D$. A quick computation shows that $\varphi$ is linear and in addition 
$$\mathcal{C}_D=\text{Im}(\varphi_D).$$
We can notice that $\varphi$ is injective, in fact $\varphi_D=0$ if and only if $\text{Im}(\varphi_D)=0$ if and only if $\mathcal{C}_D=0$ if and only if $\text{dim}_{\mathbb{F}_p}\langle d_1,\dots, d_n\rangle =0$ thanks to lemma \ref{lemma: dimension of codes in the second generic construction}, which proves our statement.\\
Also, the domain and the codomain of $\varphi$ have the same dimension of $mn$ over $\mathbb{F}_p$, hence $\varphi$ is an isomorphism. 
\\\\
If $\mathcal{C}$ is a $[n,k,d]_{\mathbb{F}_p}$ linear code, with $k\le m$, then surely there exists a linear mapping $\alpha\in\text{Hom}_{\mathbb{F}_p}(\mathbb{F}_q,\mathbb{F}_p^n)$ such that $\text{Im}(\alpha)=\mathcal{C}$ and, thanks to the bijection of $\varphi$, there exists $D\in\mathbb{F}_q^n$ such that $\varphi(D)=\alpha$, which means that 
$$\mathcal{C}=\mathcal{C}_D.$$
Notice that if $m<k$, then $\mathcal{C}$ can not be a second generic construction code thanks to lemma \ref{lemma: dimension of codes in the second generic construction}.
\end{proof}
\begin{remark}
    With the last proposition, we also give an answer to the dual problem. In particular, for a linear code $\mathcal{C}_D$ of dimension $k$, the dual code $\mathcal{C}_D^\perp$ can be made explicit with another defining set $D'$ if and only if 
    $$m\ge n-k.$$
\end{remark}
\begin{remark}
    It is important to point out how the remark $\ref{remark: generator matrix second generic construction}$ suggests how to determine a defining set that makes explicit a general linear code.\\
    Let $\mathcal{C}$ be a general $[n,k,d]$ linear code over $\mathbb{F}_p$ with $k\le m$ and the following generator matrix
    $$G=(g_1,\dots,g_n),$$
    where $g_i$ represent the columns. Let $\alpha_1,\dots,\alpha_k$ be $k$ linear independent elements of $\mathbb{F}_{q}$ over $\mathbb{F}_p$ and let $d_i$ be the elements of $\mathbb{F}_q$ represented by the coordinates of the vectors $g_i$ over the basis $\{\alpha_1,\dots,\alpha_k\}$. Then, for $D:=\{d_1,\dots,d_n\}$, we have 
    $$\mathcal{C}=\mathcal{C}_D.$$
\end{remark}


\subsection{The Walsh transform}

In this paragraph we would like to give a necessary condition for a codeword to be in the dual of a code built with the second generic construction, especially when the functions involved are weakly regular bent.
\\\\
Consider a function $f:\mathbb{F}_q\rightarrow \mathbb{F}_q$, $q=p^m$, which we will use the define the code in the second generic construction, with $p$ odd, using the deining set
$$D(f)=\{f(x)\,|\,x\in\mathbb{F}_q\}\setminus \{0\}=\{f(x_1),\dots,f(x_n)\}.$$
Now we define the following function 
$$\begin{aligned}
    g:\mathbb{F}_q &\longrightarrow \mathbb{F}_p\\ x &\longmapsto \text{Tr}_{q/p}(f(x)).
\end{aligned}$$
If we suppose that $g$ is weakly regular bent, as we already saw then there exists another function $g^*:\mathbb{F}_q\rightarrow\mathbb{F}_p$ such that, for $b\in\mathbb{F}_q$:
$$\hat{\chi}_{g}(b)=\epsilon \sqrt{p^*}^m\zeta^{g^*(b)}$$
and 
$$\hat{\chi}_{g^*}(b)=\frac{\epsilon p^m}{\sqrt{p^*}^m}\zeta^{g(x)},$$
where $\epsilon=\pm 1$ is the sign of the Walsh transform of $f(x)$, $p^*=(\frac{-1}{p})p$ and $\zeta=e^{\frac{2\pi i}{p}}$ is the primitive $p$-th root of unity. In this case, with the same notation, we have the following necessary condition, analogous to the one in the first generic construction. 
\begin{prop}\label{prop: necessary condition in the second generic construction for weakly regular bent functions}
    Let $f:\mathbb{F}_q\rightarrow\mathbb{F}_q$ be a function such that the previously defined function $g$ is weakly regular bent and consider $(c_1,\dots,c_n)\in\mathcal{C}_{D(f)}^\perp$. Then
    \begin{enumerate}
        \item if $f$ respects the scalar multiplication:
        $$\prod _{i=1}^n \hat{\chi}_{g^*}(c_ix_i)=\big(\frac{p^m}{\epsilon\sqrt{p^*}^m}\big )^n,$$
        or, equivalently,
        $$\mathfrak{Im}(\prod _{i=1}^n \hat{\chi}_{g^*}(c_ix_i))=0;$$
        \item for a generic $f$:
        $$\prod _{i=1}^n \big(\hat{\chi}_{g^*}(x_i)\big)^{c_i}=\prod_{i=1}^n\big(\frac{p^m}{\epsilon\sqrt{p^*}^m}\big )^{c_i},$$
        or, equivalently, 
        $$\mathfrak{Im}(\prod _{i=1}^n \big(\hat{\chi}_{g^*}(x_i)\big)^{c_i})=0.$$
    \end{enumerate}
\end{prop}
\begin{proof}
    If $(c_1,\dots,c_n)\in\mathcal{C}_{D(f)}^\perp$ then, from \cref{prop: the dual code in the second generic construction} we get that 
    $$\sum_{i=1}^n c_if(x_i)=0.$$
    Now we distinguish the two cases.
    \begin{enumerate}
        \item If $f$ respects the scalar multiplication, then
        $$\frac{\epsilon\sqrt{p^*}^m}{p^m}\hat{\chi}_{g^*}(c_ix_i)=\zeta^{\text{Tr}_{q/p}(f(c_ix_i))}=\zeta^{\text{Tr}_{q/p}(c_if(x_i))}.$$
        Hence, by considering the product over $i$ from $1$ to $n$, we obtain 
        $$\prod_{i=1}^n\frac{\epsilon\sqrt{p^*}^m}{p^m}\hat{\chi}_{g^*}(c_ix_i)=\zeta^{\text{Tr}_{q/p}(\sum_{i=1}^nc_if(x_i)}=\zeta^0=1,$$
        which is equivalent to our thesis and also to the fact of the imaginary part of the product being $0$ because of the structure of the Walsh transform with a fixed coefficients multiplied by a $p$-th root of the unity ($p$ is always assumed to be odd).
        \item For a generic $f$ we have 
        $$\big(\frac{\epsilon\sqrt{p^*}^m}{p^m}\hat{\chi}_{g^*}(x_i)\big)^{c_i}=\big(\zeta^{\text{Tr}_{q/p}(f(x_i))}\big)^{c_i}=\zeta^{\text{Tr}_{q/p}(c_if(x_i))}.$$
        Hence, by considering the product over $i$ from $1$ to $n$, we obtain 
        $$\prod_{i=1}^n\big(\frac{\epsilon\sqrt{p^*}^m}{p^m}\hat{\chi}_{g^*}(x_i)\big)^{c_i}=\zeta^{\text{Tr}_{q/p}(\sum_{i=1}^nc_if(x_i) }=\zeta^0=1,$$
        which is equivalent to our thesis and also to the one involving the imaginary part for the same reasons explained in point $1$.
    \end{enumerate}
\end{proof}
After considering a weakly regular bent function, now we would like to establish a connection between the dual code in the second generic construction and the Walsh transform in the general case. 
\\\\
Consider a function $f:\mathbb{F}_q\rightarrow\mathbb{F}_q$, $q=p^m$, which we will use to define the code in the second generic construction, and for every $i=1,\dots,n$ define the functions $g_i:\mathbb{F}_q\rightarrow\mathbb{F}_p$ such that $g_i(x_i)=\text{Tr}_{q/p}(f(x_i)+x_i)$ and $g_i(x)=\text{Tr}_{q/p}(x)$ for $x\ne x_i$. Then we have the following necessary condition.
\begin{prop}\label{prop: necessary condition in the second general construction for the general case}
    Let $f:\mathbb{F}_q\rightarrow\mathbb{F}_q$ be a function, $g_i:\mathbb{F}_q\rightarrow\mathbb{F}_p$ be the previously defined functions and consider $(c_1,\dots,c_n)\in\mathcal{C}_{D(f)}^\perp$. Then 
    $$\prod_{i=1}^n\big(\hat{\chi}_{g_i}(1)+1-q\big)^{c_i}=1.$$
\end{prop}
\begin{proof}
    As we already observed, from \ref{prop: the dual code in the second generic construction} we have that 
$$\sum_{i=1}^n c_if(x_i)=0.$$
If we compute the first Walsh-Hadamard coefficient of the function $g_i$ we get the following result:
$$\hat{\chi}_{g_i}(1)=\sum_{x\in\mathbb{F}_q}\zeta^{g_i(x)-\text{Tr}_{q/p}(x)}=q-1+\zeta^{\text{Tr}_{q/p}(f(x_i))},$$
and so 
$$\big(\hat{\chi}_{g_i}(1)+1-q\big)^{c_i}=\zeta^{\text{Tr}_{q/p}(c_if(x_i))}.$$
Hence we obtain our thesis:
$$\prod_{i=1}^n\big(\hat{\chi}_{g_i}(1)+1-q\big)^{c_i}=\zeta^{\text{Tr}_{q/p}(\sum_{i=1}^nc_if(x_i))}=\zeta^0=1.$$
\end{proof}
From this necessary condition, also for the second generic construction, with an analogous proof, we can obtain the following restriction on the computation of the dual and also a precise computation when the dimension of the code is equal to one, using character theory.
\begin{prop}\label{prop: dual of the second generic construction and character theory}
    Consider the functions $g_i:\mathbb{F}_q\rightarrow\mathbb{F}_p$ of theorems \ref{prop: necessary condition in the second general construction for the general case}, then the function
    $$\begin{aligned}
    \varphi:\mathbb{F}_p^n &\longrightarrow \mathbb{C}\\ (c_1,\dots,c_n) &\longmapsto \prod_{i=1}^n\big(\hat{\chi}_{g_i}(1)+1-q\big)^{c_i}.
\end{aligned}$$
is an additive character of $\mathbb{F}_{p^n}$ after the identification of $\mathbb{F}_p^n$ with $\mathbb{F}_{p^n}$ such that 
$$\mathcal{C}_{D(f)}^\perp\subseteq\text{ker}(\varphi).$$
In particular, if $\text{dim}(\mathcal{C}_{D(f)})=1$, then
$$\mathcal{C}_{D(f)}^\perp=\text{ker}(\varphi).$$
\end{prop}

\subsection{The Hull code}

Thanks to the study of the dual code in the second generic construction, we are able to characterize the Hull code in this case in the following explicit way. 

\begin{prop}\label{prop: hull code in the second generic construction}
    Let $D=\{d_1,\dots,d_n\}\subseteq\mathbb{F}_q$ be a defining set for a code in the second generic construction, and let $\mathcal{L}$ be the linear code over $\mathbb{F}_q$ generated by the codeword $l=(d_1,\dots,d_n)$. Then
$$\text{Hull}(\mathcal{C}_D)= \mathcal{L}^\perp\cap \mathcal{C}_D.$$
\end{prop}
\begin{proof}
    The assertion comes directly from theorem \ref{prop: the dual code in the second generic construction}, in fact 
    $$\text{Hull}(\mathcal{C}_D)=\mathcal{C}_D^\perp\cap \mathcal{C}_D= \mathcal{L}^\perp\cap\mathbb{F}_p^n\cap \mathcal{C}_D=\mathcal{L}^\perp\cap \mathcal{C}_D.$$
\end{proof}
For this reason, we can characterize the Hull code also in the following way.
\begin{prop}\label{prop: hull code in the second generic construction via kernel}
    Let $D=\{d_1,\dots,d_n\}\subseteq \mathbb{F}_q$ be a defining set for a code in the second generic construction. Define $\textbf{c}_x:=(\text{Tr}_{q_p}(xd))_{d\in D}$
    and consider the following linear mapping
    $$\begin{aligned}
    \varphi:\mathcal{C}_D &\longrightarrow \mathbb{F}_q\\ \textbf{c}_x &\longmapsto \sum_{d\in D}\text{Tr}(xd)d.
\end{aligned}$$
Then 
$$\text{Hull}(\mathcal{C}_D)=\text{ker}(\varphi)$$
and in particular, if $\text{dim}(\mathcal{C}_D)=k$ then 
$$\text{dim}(\text{Hull}(\mathcal{C}_D))=l\,\,\iff \,\, \text{rk}(\varphi)=k-l.$$
\end{prop}
\begin{proof}
    A quick computation shows that the mapping is actually linear and from \ref{prop: hull code in the second generic construction}, we have that $\textbf{c}_x\in\text{Hull}(\mathcal{C}_D)$ if and only if
    $$\sum_{d\in D}\text{Tr}_{q/p}(xd)d=0,$$
so this let us conclude that $\text{Hull}(\mathcal{C}_D)=\text{ker}(\varphi)$ and also the assertion on the dimension of the Hull code.
\end{proof}
Also in this case, considering that a codeword in the Hull code is actually in the dual as well, we can restate all the propositions in section 5 to obtain a necessary condition for a codeword in $\mathcal{C}_D$ to be in the Hull; as an example, we propose the first proposition, the others can be obtained in the same way.
\\\\
Consider a function $f:\mathbb{F}_q\rightarrow \mathbb{F}_q$, $q=p^m$, which we will use the define the code in the second generic construction, with $p$ odd, using the defining set
$$D(f)=\{f(x)\,|\,x\in\mathbb{F}_q\}\setminus \{0\}=\{f(x_1),\dots,f(x_n)\}.$$
Now we define the following function 
$$\begin{aligned}
    g:\mathbb{F}_q &\longrightarrow \mathbb{F}_p\\ x &\longmapsto \text{Tr}_{q/p}(f(x)).
\end{aligned}$$
If we suppose that $g$ is weakly regular bent, as we already saw then there exists another function $g^*:\mathbb{F}_q\rightarrow\mathbb{F}_p$ such that, for $b\in\mathbb{F}_q$:
$$\hat{\chi}_{g}(b)=\epsilon \sqrt{p^*}^m\zeta^{g^*(b)}$$
and 
$$\hat{\chi}_{g^*}(b)=\frac{\epsilon p^m}{\sqrt{p^*}^m}\zeta^{g(x)},$$
where $\epsilon=\pm 1$ is the sign of the Walsh transform of $f(x)$, $p^*=(\frac{-1}{p})p$ and $\zeta=e^{\frac{2\pi i}{p}}$ is the primitive $p$-th root of unity. In this case, with the same notation, we have the following necessary condition, analogous to the one in the first generic construction. 
\begin{prop}\label{prop: hull necessary condition in the second generic construction for weakly regular bent functions}
    Let $f:\mathbb{F}_q\rightarrow\mathbb{F}_q$ be a function such that the previously defined function $g$ is weakly regular bent and consider $\textbf{c}_x\in\mathcal{C}_{D(f)}$. If $\textbf{c}_x\in\text{Hull}(\mathcal{C}_{D(f)})$ then 
    \begin{enumerate}
        \item if $f$ respects the scalar multiplication:
        $$\prod _{i=1}^n \hat{\chi}_{g^*}(\text{Tr}(xf(x_i))x_i)=\big(\frac{p^m}{\epsilon\sqrt{p^*}^m}\big )^n,$$
        or, equivalently,
        $$\mathfrak{Im}(\prod _{i=1}^n \hat{\chi}_{g^*}(\text{Tr}(xf(x_i))x_i))=0;$$
        \item for a generic $f$:
        $$\prod _{i=1}^n \big(\hat{\chi}_{g^*}(x_i)\big)^{\text{Tr}(xf(x_i))}=\prod_{i=1}^n\big(\frac{p^m}{\epsilon\sqrt{p^*}^m}\big )^{\text{Tr}(xf(x_i))},$$
        or, equivalently, 
        $$\mathfrak{Im}(\prod _{i=1}^n \big(\hat{\chi}_{g^*}(x_i)\big)^{\text{Tr}(xf(x_i))})=0.$$
    \end{enumerate}
\end{prop}

From these necessary condition, also for the second generic construction we can obtain the following restriction on the computation of the Hull, using character theory.
\begin{prop}\label{prop: hull dual of the second generic construction and character theory}
    Consider the functions $g_i:\mathbb{F}_q\rightarrow\mathbb{F}_p$ of \ref{prop: necessary condition in the second general construction for the general case}, then the function
    $$\begin{aligned}
    \varphi:\mathcal{C}_{D(f)} &\longrightarrow \mathbb{C}\\ \textbf{c}_x &\longmapsto \prod_{i=1}^n\big(\hat{\chi}_{g_i}(1)+1-q\big)^{\text{Tr}(xf(x_i))}.
\end{aligned}$$
is an additive character of $\mathcal{C}_{D(f)}$, after the usual identifications, such that 
$$\text{Hull}(\mathcal{C}_{D(f)})\subseteq\text{ker}(\varphi).$$
\end{prop}

\subsection{Construction of codes of fixed Hull dimension}

In this section we would like to present a construction method of linear codes of fixed Hull dimension, as an application of the Hull characterizations just seen. 
\\\\
Let $\mathbb{F}_p$ be a finite field in which $-1$ is a quadratic residue (for example $p=5$) and in particular let $\alpha\in\mathbb{F}_p$ be a root of the polynomial $x^2+1$. Also consider $\beta \in\mathbb{F}_p$ such that $\beta^2\ne -1$ and take $d_1,\dots,d_k\in\mathbb{F}_q$, $q=p^m$, which are linear independent over $\mathbb{F}_p$ and $0\le l\le k$.
\\
Now we consider the following defining set
$$D=\{d_1,\dots,d_k,\alpha d_1,\dots, \alpha d_l, \beta d_{l+1},\dots,\beta_{d_k}\}.$$
\begin{lemma}
    The linear code $\mathcal{C}_D$ built with the second generic construction is a $[2k,k,2]$ linear code over $\mathbb{F}_p$ of Hull dimension $l$.
\end{lemma}
\begin{proof}
    Since $|D|=2k$, the length of the code is equal to $2k$; also, its dimension is equal to $k$ thanks to lemma \ref{lemma: dimension of codes in the second generic construction}.
    \\\\
    For $x\in\mathbb{F}_q$, we define $\textbf{c}_x:=(\text{Tr}_{q/p}(xd))_{d\in D}$; then theorem \ref{prop: hull code in the second generic construction} tells us that $\textbf{c}_x\in\text{Hull}(\mathcal{C}_D)$ if and only if
    $$\sum_{d\in D}\text{Tr}_{q/p}(xd)d=0.$$
    Thanks to the linear independence of the elements in $D$, we have that $\textbf{c}_x\in\text{Hull}(\mathcal{C}_D)$ if and only if
    $$
    \begin{cases}
  \text{Tr}(xd_1)d_1+\text{Tr}(x\alpha d_1)\alpha d_1=0 \\
       \text{Tr}(xd_2)d_2+\text{Tr}(x\alpha d_2)\alpha d_2=0 \\
       \dots \\
       \text{Tr}(xd_l)d_l+\text{Tr}(x\alpha d_l)\alpha d_l=0 \\
       \text{Tr}(xd_{l+1})d_{l+1}+\text{Tr}(x\beta d_{l+1})\beta d_{l+1}=0 \\
       \dots \\
       \text{Tr}(xd_k)d_k+\text{Tr}(x\beta d_k)\beta d_k=0
\end{cases}
\iff 
\begin{cases}
       (\alpha^2+1)\text{Tr}(xd_1)d_1=0 \\
       (\alpha^2+1)\text{Tr}(xd_2)d_2=0 \\
       \dots \\
       (\alpha^2+1)\text{Tr}(xd_l)d_l=0 \\
       (\beta^2+1)\text{Tr}(xd_{l+1})d_{l+1}=0 \\
       \dots \\
       (\beta^2+1)\text{Tr}(xd_k)d_k=0.
\end{cases}
$$
Hence if and only if 
$$\text{Tr}(xd_i)=0$$
for every $i=l+1,\dots,k$, so if and only if
$$x\in\bigcap_{i=l+1}^kd_i^{-1}\text{ker}(\text{Tr}_{q/p}).$$
If we consider the mapping $\psi:\mathbb{F}_q\rightarrow \mathcal{C}_D$ such that $\psi(x):=\textbf{c}_x$, then we have the identification
$$\frac{\mathbb{F}_q}{\text{ker}(\psi)}\cong \mathcal{C}_D,$$
where
$$\text{ker}(\psi)=\bigcap_{d\in D}d^{-1}\text{ker}(\text{Tr}_{q/p}).$$
Now $\text{dim}(\mathcal{C}_D)=k$ and $\text{dim}(\mathbb{F}_q)=m$, then 
$$\text{dim}\big (\bigcap_{d\in D}d^{-1}\text{ker}(\text{Tr}_{q/p})\big )=m-k.$$
With this identification, we have that 
$$\text{Hull}(\mathcal{C}_D)\cong \frac{\bigcap_{i=l+1}^kd_i^{-1}\text{ker}(\text{Tr}_{q/p})}{\bigcap_{d\in D}d_i^{-1}\text{ker}(\text{Tr}_{q/p})}.$$
Also, $\text{dim}(\bigcap_{i=l+1}^kd_i^{-1}\text{ker}(\text{Tr}_{q/p}))=m+k-l$ (to show this, it suffices to consider the defining set $D'=\{d_{l+1},\dots,d_k\}$ and apply lemma \ref{lemma: dimension of codes in the second generic construction}). Hence $\text{dim}_{\mathbb{F}_p}(\text{Hull}(\mathcal{C}_D))=m+l-k-(m-k)=l$, as we wanted to show.
\\\\
For the minimal distance, we can observe that, thanks to lemma \ref{lemma: dimension of codes in the second generic construction}, from the first $k$ entries of the code we get a code isomorphic to $\mathbb{F}_p^k$; then the other entries are multiplied for $\alpha$ or $\beta$. Hence, we can conclude that each weight is even and in particular that the minimal distance is equal to 2.
\end{proof}
\begin{remark}
    From the latter observation in the proof, we can actually deduce the complete weight enumerator polynomial of the code, which is the following.
    $$p(x)=1+(p-1)\binom{k}{1}x^2+(p-1)^2\binom{k}{2}x^4+\dots+(p-1)^s\binom{k}{s}x^{2s}+\dots+(p-1)^k\binom{k}{k}x^{2k}.$$
\end{remark}
\begin{remark}
    Since the construction method gives linear codes of any Hull dimension, we can especially choose some small values of $l$ in order to get some small Hull codes, on which research in coding theory has recently been focusing because of the faster algorithms in the context of the permutations groups and of the numerous applications in quantum coding. 
\end{remark}

\subsection{An example of LCD code}
As an application of the Hull characterization previously shown, in this section we would like to present a construction method of LCD codes via the second generic construction method over a field of even characteristic. 
\\\\
Let $q=2^a$ and $r=q^b$ with $a,b\in\mathbb{N}^*$; fix $k$ linearly independent elements $d_1,\dots,d_k$ of $\mathbb{F}_r$ over $\mathbb{F}_q$, with $k$ even and define the set $D_1=\{d_1,\dots,d_k\}$. Let $D_2$ be the set of elements in $D_1$, for example
$$D_2=\{d_1+d_2,d_3+d_4,\dots,d_{k-1}+d_k\}.$$
We consider the defining set $D=D_1\cup D_2$ and we define the code
$$\mathcal{C}_D=\{(\text{Tr}_{r/q}(xd)_{d\in D})\,|\, x\in\mathbb{F}_r\}.$$
\begin{lemma}
    The obtain code $\mathcal{C}_D$ is LCD of length $\frac{3k}{2}$ and dimension $k$.
\end{lemma}
\begin{proof}
    The length of the code is equal to $\frac{3k}{2}$ since $|D|=\frac{3k}{2}$ and its dimension is equal to $k$ thanks to lemma \ref{lemma: dimension of codes in the second generic construction}. To study the Hull of the code, we can use theorem \ref{prop: hull code in the second generic construction via kernel} and we compute
    $$\text{Tr}(x(d_1+d_2))(d_1+d_2)=\text{Tr}(xd_1)d_1+\text{Tr}(xd_2)d_2+\text{Tr}(xd_1)d_2+\text{Tr}(xd_2)d_1.$$
    Hence, thanks to the even characteristic, we have that the first $k$ rows of the matrix associated to $\varphi$ is permutationally equivalent to the transpose of the generator matrix of $\mathcal{C}_D$, which means that $\text{rk}(\varphi)=k$ and so 
    $$\text{dim}(\text{Hull}(\mathcal{C}_D))=k-\text{rk}(\varphi)=0.$$
\end{proof}
\begin{remark}
    If we consider theorem \ref{prop: the dual code in the second generic construction}, we can explore the dual of the code, in fact $(c_1,\dots,c_{\frac{3k}{2}})\in\mathbb{F}_p^{\frac{3k}{2}}$ is in the dual code if and only if
    $$\sum_{i=1}^{\frac{3k}{2}}c_id_i$$
    and thanks to the linear independence of the elements in $D$ we deduce that the dual code is determined by the last $k/2$ entries of the code, that its minimal distance is $3$, that every weight is divisible by 3 and that the complete weight enumerator polynomial is the following
    $$p(x)=1+(q-1)\binom{\frac{k}{2}}{1}x^3+(q-1)^2\binom{\frac{k}{2}}{2}x^6+\dots+(q-1)^i\binom{\frac{k}{2}}{i}x^{3i}+\dots+(q-1)^{\frac{k}{2}}x^{\frac{3k}{2}}.$$
\end{remark}
\begin{remark}
    We can also observe that the length of the code might be reduced with a smaller set $D_2$ and that, once $k$ is fixed, there are a priori
    $$\frac{k!}{(k/2)!2^{k/2}}$$
    linear codes of this shape.
\end{remark}

\section{The weight enumerator in the case of even characteristic}

In the case of even characteristic, we recall that we have no difference between bent and weakly regular bent functions, since in this case the image of the Walsh transform is real, so the dual function indicates just the sign of the transform. 
\\\\
For these reasons, the thesis of \ref{prop: first necessary condition in the first generic construction for weakly regular bent functions} and \ref{prop: second necessary condition in the first generic construction for weakly regular bent functions}, in the case $f$ respects the scalar multiplication, can be restated in the following way, in the same hypothesis:
        $$\prod _{i=1}^q \hat{\chi}_{g^*}(c_ix_i)>0;$$
On the other hand, for the second generic construction, the thesis of \ref{prop: necessary condition in the second generic construction for weakly regular bent functions}, always in the case $f$ respects the scalar multiplication, can be restated in the following way, in the same hypothesis:
        $$\prod _{i=1}^n \hat{\chi}_{g^*}(c_ix_i)>0.$$
In the even characteristic, the seen necessary conditions can help us to compute the weight of the codewords in relation with the Walsh transforms of the previously cited functions, in a similar way as Mesnager did in \cite{WeightWalshTransform} for the codewords in the actual first generic construction code. In particular, we have the following assertion in the first generic construction. 
\begin{prop}\label{prop: weight enumerator and Walsh transform in the first generic construction}
    Let $f:\mathbb{F}_q\rightarrow\mathbb{F}_q$ be a function such that one of the two functions $g$ defined in \ref{prop: first necessary condition in the first generic construction for weakly regular bent functions} or \ref{prop: second necessary condition in the first generic construction for weakly regular bent functions} is (weakly regular) bent and consider $\textbf{c}=(c_1,\dots,c_q)\in\mathcal{C}(f)^\perp$. Then, if 
    $$\alpha:= \prod _{i=1}^q \big(\hat{\chi}_{g^*}(x_i)\big)^{c_i}$$
    we have 
    $$\text{wt}(\textbf{c})=\frac{\operatorname{log}_2(\alpha^2)}{m}.$$
\end{prop}
\begin{proof}
    From \ref{prop: first necessary condition in the first generic construction for weakly regular bent functions} and \ref{prop: second necessary condition in the first generic construction for weakly regular bent functions} we have that 
    $$\alpha=\prod_{i=1}^q (\epsilon 2^{\frac{m}{2}})^{c_i}=(\epsilon 2^{\frac{m}{2}})^{\text{wt}(\textbf{c})},$$
    hence
    $$|\alpha|=(2^\frac{m}{2})^{\text{wt}(\textbf{c})};$$
    $$\operatorname{log}_2(|\alpha|)=\frac{m\text{wt}(\textbf{c})}{2},$$
    from which we derive the thesis.
\end{proof}
We have a similar assertion in the case of the second generic construction.
\begin{prop} \label{prop: weight enumerator and Walsh transform in the second generic construction}
    Let $f:\mathbb{F}_q\rightarrow\mathbb{F}_q$ be a function such that the function $g$ defined in \ref{prop: necessary condition in the second generic construction for weakly regular bent functions} is (weakly regular) bent and consider $\textbf{c}=(c_1,\dots,c_n)\in\mathcal{C}_{D(f)}^\perp$. Then, if 
    $$\alpha:= \prod _{i=1}^n \big(\hat{\chi}_{g^*}(x_i)\big)^{c_i}$$
    we have 
    $$\text{wt}(\textbf{c})=\frac{\operatorname{log}_2(\alpha^2)}{m}.$$
\end{prop}
\begin{proof}
    The same as in \ref{prop: weight enumerator and Walsh transform in the first generic construction}, making use of \ref{prop: necessary condition in the second generic construction for weakly regular bent functions}.
\end{proof}

\section*{Conclusion and remarks}
After the explicit characterization of the dual codes in the first two generic constructions, this paper confirms the deep relation between the fundamental concept of the Walsh transform and the linear codes built by the mean of cryptographic functions. In fact, as we saw, the dual codes are strictly connected to the Walsh-Hadamard coefficients of the cryptographic functions involved, which can help to explicitly compute the dual codes or to build codes with fixed dual parameters.
\\\\
It would be interesting to find some sufficient condition for a codeword to be in the dual codes, for example the necessary conditions presented tells us that some of the expressions
$$\sum_{i=1}^q c_ix_i\,\,\text{or}\,\,\sum_{i=1}^qc_if(x_i)$$
are in the kernel of the trace map, hence it suffices to make sure that those expressions are null, for example by saying 
$$\emph{there exists no}\,\,\alpha\in\mathbb{F}_q^*\,\,\emph{such that}\,\,\sum_{i=1}^q c_ix_i=\alpha^p-\alpha,$$
making use of the characterization of the trace map that we can find in \cite{lidl1994introduction}.
\\\\
Some future research paths after this paper could focus on: some other sufficient conditions for a codeword to be in the dual code, some other weight enumerators (always involving the Walsh transform) and applications to the associated Hull code as a particular subcode of the dual code (and specific conditions for which the Hull code happens to be low dimensional).

\section*{Acknowledgements}
The author sincerely thanks Sihem Mesnager for her valuable support and her very beneficial review of this work. 

%
%

\end{document}